\documentclass[letterpaper, 10 pt, conference]{ieeeconf}  
\IEEEoverridecommandlockouts
\usepackage{balance}
\usepackage{color}
\usepackage{caption}
\usepackage{subcaption}
\usepackage{enumerate}
\usepackage{cite}
\usepackage{empheq}
\usepackage{mathrsfs}
\usepackage{flushend}
\usepackage{mathtools}

\usepackage{tikz}
\usepackage{circuitikz}


\makeatletter
\pgfcircdeclarebipole{}{\ctikzvalof{bipoles/interr/height 2}}{spst}{\ctikzvalof{bipoles/interr/height}}{\ctikzvalof{bipoles/interr/width}}{

    \pgfsetlinewidth{\pgfkeysvalueof{/tikz/circuitikz/bipoles/thickness}\pgfstartlinewidth}

    \pgfpathmoveto{\pgfpoint{\pgf@circ@res@left}{0pt}}
    \pgfpathlineto{\pgfpoint{.6\pgf@circ@res@right}{\pgf@circ@res@up}}
    \pgfusepath{draw}   
}

\def\pgf@circ@spst@path#1{\pgf@circ@bipole@path{spst}{#1}}
\tikzset{switch/.style = {\circuitikzbasekey, /tikz/to path=\pgf@circ@spst@path, l=#1}}
\tikzset{spst/.style = {switch = #1}}
\makeatother

\makeatletter
\let\proof\@undefined                        
\let\endproof\@undefined                  
\makeatother
\usepackage{graphicx,amssymb,amstext,amsmath,amsthm}

\usepackage[bookmarks=true]{hyperref}
\usepackage{algorithm,algorithmicx,algpseudocode}
\algnewcommand{\algorithmicgoto}{\textbf{go to}}%
\algnewcommand{\Goto}[1]{\algorithmicgoto~\ref{#1}}%
\algnewcommand{\LineComment}[1]{\Statex \(\triangleright\) #1}
\algnewcommand{\LineCommentN}[1]{\Statex \hspace{1cm}\(\triangleright\) #1}

\usepackage{multirow}
\usepackage{stfloats}

\newtheorem{prop}{Proposition} 
\newtheorem{cor}{Corollary}
\newtheorem{thm}{Theorem}
	\newtheorem{assumption}{Assumption}
\newtheorem{lem}{Lemma}
\newtheorem{defn}{Definition}

\newtheorem{problem}{Problem}

\usepackage{setspace}
\let\oldbibliography\thebibliography
\renewcommand{\thebibliography}[1]{%
  \oldbibliography{#1}%
}

\definecolor{darkgreen}{rgb}{0.0, 0.5, 0.0}

\newcommand{\yong}[1]{{\color{black} #1}}
\newcommand{\moh}[1]{{\color{black} #1}}
\newcommand{\yongn}[1]{{\color{black} #1}}
\newcommand{\yongs}[1]{{\color{black} #1}}
\newcommand{\moha}[1]{{\color{black} #1}}
\newcommand{\moham}[1]{{\color{black} #1}}
\newcommand{\mo}[1]{{\color{black} #1}}
\newcommand{\yo}[1]{{\color{black} #1}}
\newcommand{\mk}[1]{{\color{black} #1}}
\newcommand{\sy}[1]{{\color{black} #1}}
\newcommand{\mj}[1]{{\color{black} #1}}
\allowdisplaybreaks

\begin{document}

\title{\LARGE \bf 
Resilient State Estimation for Nonlinear Discrete-Time Systems via Input and State Interval Observer Synthesis}

\author{%
Mohammad Khajenejad, Zeyuan Jin, Thach Ngoc Dinh and Sze Zheng Yong\\
\thanks{
M. Khajenejad is with the 
University of California, San Diego, 
CA, USA. Z. Jin is with 
Arizona State University, Tempe, AZ, USA. T.N. Dinh is with Conservatoire National des Arts et Métiers (CNAM), CEDRIC-Laetitia, 
Paris, 
France. S.Z. Yong is with 
Northeastern University, Boston,
MA, USA. (e-mail: mkhajenejad@ucsd.edu, zjin43@asu.edu, ngoc-thach.dinh@lecnam.net, s.yong@northeastern.edu).}
\thanks{This work is partially supported by NSF grant CNS-2313814.}
}

\maketitle
\thispagestyle{empty}
\pagestyle{empty}

\begin{abstract}
This paper addresses the problem of \mo{resilient state estimation and attack reconstruction for bounded-error nonlinear discrete-time systems with nonlinear observations/constraints, where both sensors and actuators \sy{can be} compromised \yo{by false data injection} 
attack signals/unknown inputs}. 
\mo{By leveraging} mixed-monotone decomposition of nonlinear \mo{functions, as well as} affine parallel outer-approximation of the observation functions, \mo{along with} introducing auxiliary states \mo{to cancel out the effect of the attacks/unknown inputs}, 
our proposed observer recursively computes interval estimates that by construction, contain the true states and unknown inputs of the system. Moreover, we provide several semi-definite programs 
to synthesize observer gains \mo{to ensure} input-to-state \mo{stability of the proposed observer} and \mo{optimality of} the design in the sense of \yo{minimum  $\mathcal{H}_{\infty}$ gain}.
 \end{abstract}
\section{Introduction} 

\yo{State estimation and unknown input reconstruction are indispensable in various} 
engineering applications such as aircraft tracking, fault detection, attack detection and mitigation in cyber-physical systems (CPS) and urban transportation \cite{liu2011robust,yong2016tcps,yong2018simultaneous}. \mo{Particularly, 
set-membership approaches have been proposed \yo{for bounded-error systems} to provide hard
accuracy bounds, which is especially useful for obtaining robustness guarantees \mj{for} safety-critical
systems. Moreover, since 
attackers may be
strategic in adversarial settings, the ability to simultaneously
estimate states and inputs without imposing any assumptions
on the unknown inputs\yo{/attack signals} is desirable and often crucial.}

\emph{Literature review.} 
Numerous studies in the literature have investigated \mo{\emph{secure estimation}, i.e., how to accurately estimate the states of a system when it is under attack or subject to adversarial signals. For instance, {secure state estimation} and control problem was addressed in the presence of \yo{false data injection attacks on both the actuators and sensors} 
in \cite{chen2021resilient}, in which a $\chi^2$ detector was proposed to detect malicious attacks. 
The research in \cite{wu2018secure} proposed a sliding-mode observer to simultaneously estimate system states and attacks, while the work in \cite{mousavinejad2018novel} provided a projected sliding-mode observer-based estimation approach to reconstruct system states. \yo{Further, the work in \cite{corradini2017robust}  
\yo{reconstructed} attack signals from the
equivalent output injection signal using a sliding-mode observer, while in \cite{lu2017secure}, an attack was considered as an auxiliary state and estimated by employing a robust switching Luenberger observer assuming sparsity.}
However, all the aforementioned works considered stochastic/Gaussian noise and hence do not apply to \yo{the} bounded-error setting \yo{we consider in this paper}, 
where noise/disturbance signals are assumed to be distribution-free and bounded.} 

\yo{A related body of literature that could be applied to resilient state estimation in the bounded-error setting is that of unknown input interval observers. Particularly,} 
the works in \cite{ellero2019unknown,wang2021novel,marouani2021unknown} considered the problem of designing 
unknown input interval observer\yo{s} for continuous-time linear parameter varying (LPV), uncertain linear time-invariant (LTI) and discrete-time switched linear systems, respectively, where  the authors \yo{in \cite{ellero2019unknown}} formulated the necessary Metzler property as part of a semi-definite program. A similar problem was considered for nonlinear continuous-time systems with linear observations in \cite{wei2021hybrid}. 
However, \yo{these approaches are} 
not suitable for general discrete-time nonlinear systems \yo{and} 
the unknown input signal \mo{does not affect the output/measurement equation} \yo{(needed for representing false data injection attacks on the sensors)} in either of \yo{the works in} \cite{ellero2019unknown,wang2021novel,marouani2021unknown,wei2021hybrid}.

\yo{On the other hand, while}  our previous works \cite{khajenejad2020simultaneousfullrank,khajenejad2020simultaneous} \yo{do consider} the design of state and unknown input interval observers 
for nonlinear discrete-time systems with nonlinear observations, 
no stabilizing gains were synthesized in \cite{khajenejad2020simultaneousfullrank,khajenejad2020simultaneous}. \mo{We aim to address this shortcoming in this paper}. 

 \emph{Contributions.} 
By leveraging a combination of mixed-monotone decomposition \mo{of nonlinear functions} \cite{moh2022intervalACC,khajenejad2022h} and parallel affine outer-approximation of observation functions {\cite{singh2018mesh}}, we synthesize a resilient interval observer, i.e., a discrete-time dynamical system 
that by construction, \emph{simultaneously} returns interval-valued estimates of states and unknown inputs \yo{(representing false data injection signals on both the actuators and sensors)} for \moh{a broad range of nonlinear discrete-time} systems with nonlinear observations. Our proposed design is a significant improvement to our previous input and state interval observer designs \cite{khajenejad2020simultaneous,khajenejad2020simultaneousfullrank}, in which no stabilizing gains were considered and so the stability of the 
\yo{previous observer designs} only hinged upon some dynamical systems properties. Moreover, in contrast to many  
 unknown input (interval) observer designs in the literature, our design considers 
 arbitrary {unknown} input signals 
 {with no assumptions of} \moham{\emph{a priori} known intervals}, being stochastic with zero mean (as is often assumed for noise) or bounded. 
Further, we provide sufficient conditions 
for the input-to-state-stability of the proposed observer, which at the same time ensures the optimality of the design in the sense of \yo{minimum $\mathcal{H}_{\infty}$ gain} by solving semi-definite programs. 

\section{Preliminaries}
{\emph{{Notation}.}} {\color{black}$\lor$ denotes the logical disjunction (the OR truth-functional operator)}. $\mathbb{R}^n,\mathbb{R}^{n  \times p},\mathbb{D}_n,\mathbb{N},\mathbb{N}_n,{\mathbb{R}_{\geq 0}}$ {and $\mathbb{R}_{>0}$} denote the $n$-dimensional Euclidean space and the sets of $n$ by $p$ matrices, $n$ by $n$ diagonal matrices, natural numbers \yo{(including 0)}, natural numbers \yo{from 1 to $n$,} 
{{non-negative} and positive real numbers,} respectively, while $\mathbb{M}_n$ 
 denote{s} the set of all $n$ by $n$ Metzler {matrices}, {i.e., square matrices 
 {whose} off-diagonal {elements} are non-negative. {\color{black}Euclidean norm of a vector $x \in \mathbb{R}^n$ is denoted by $\|x\|_{2}\mo{\triangleq \sqrt{x^\top x}}$.}} 
 For $M \in \mathbb{R}^{n \times p}$, $M_{ij}$ denotes $M$'s entry in the $i$'th row and the $j$'th column, $M^{\oplus}\triangleq \max(M,\mathbf{0}_{n,p})$, $M^{\ominus}=M^{\oplus}-M$ and $|M|\triangleq M^{\oplus}+M^{\ominus}$, where $\mathbf{0}_{n,p}$ is the zero matrix in $\mathbb{R}^{n \times p}$, \yong{while $\textstyle{\mathrm{sgn}}(M) \in \mathbb{R}^{n \times p}$ is the element-wise sign of $M$ with $\textstyle{\mathrm{sgn}}(M_{ij})=1$ if $M_{ij} \geq 0$ and $\textstyle{\mathrm{sgn}}(M_{ij})=-1$, otherwise.} 
 {$M \succ 0$ and $M \prec 0$ (or $M \succeq 0$ and $M \preceq 0$) denote that $M$ is positive and negative (semi-)definite, respectively}. {Further, a function $f:S \subseteq \mathbb{R}^n \to \mathbb{R}$, where $0 \in S$, is positive definite if $f(x) >0$ for all $x \in S{\, \setminus \{0\}}$, and $f(0)=0$.} Finally, an interval {$\mathcal{I} \triangleq [\underline{{z}},\overline{{z}}]  \subset 
\mathbb{R}^n$} is the set of all real vectors ${z \in \mathbb{R}^{n_z}}$ that satisfies $\underline{{z}} \le {z} \le \overline{{z}}$ {(component-wise)}, where 
$\|\overline{{z}}-\underline{{z}}\|{_{\infty}\triangleq \max_{i \in \{1,\cdots,{n_z}\}}{|z_i|}}$ is the 
{interval} width of $\mathcal{I}$. 

Next, we \mk{review} some related results and definitions. 

\begin{prop}[Jacobian Sign-Stable Decomposition {\cite[Proposition 2]{moh2022intervalACC}}]\label{prop:JSS_decomp}
If a mapping $f: {\mathcal{Z}} \subset \mathbb{R}^{n_z} \to \mathbb{R}^p$ has Jacobian matrices satisfying $J^f({x}) \in [\underline{J}^f,\overline{J}^f]$, $\forall {z \in \mathcal{Z}}$, where $\underline{J}^f,\overline{J}^f \in \mathbb{R}^{p \times n_z}$ are known matrices, then the mapping $f$ can be decomposed into an additive remainder-form:
\begin{align}\label{eq:JSS_decomp}
\forall {z \in \mathcal{Z}},f({\mk{z}})=H{z}+\mu({z}),
\end{align}
 {where} the matrix $H\in\mathbb{R}^{\mk{p \times n_z}}$ {satisfies}
 \begin{align}\label{eq:H_decomp}
 \forall (i,j) \in \mathbb{N}_p \times \mathbb{N}_{n_z}, H_{ij}=\underline{J}^f_{ij} \ \lor H_{ij}=\overline{J}^f_{i,j} ,    
 \end{align}
  {and} $\mu(\cdot)$ and $H{z}$ are {nonlinear and linear} Jacobian sign-stable (JSS) mappings, {respectively,} i.e., the signs of each element of their Jacobian matrices do not change within their domains ($J^\nu_{ij}(\cdot) \geq 0 \ \text{ or } J^\nu_{ij}(\cdot) \leq 0$, $\nu(z) \in \{\mu(z),Hz\}$).
\end{prop}

\begin{defn}[Mixed-Monotonicity {and} Decomposition Functions] \cite[Definition 1]{abate2020tight},\cite[Definition 4]{yang2019sufficient} \label{defn:dec_func}
Consider the discrete-time dynamical system $x_{k+1}= g(x_k)$, 
with initial state $x_0 \in \mathcal{X}_0 \triangleq [\underline{x}_0,\overline{x}_0] {\subset \mathbb{R}^{n}}$.
{Furthermore}, ${g}:{\mathcal{X}} \subset \mathbb{R}^{n} \to \mathbb{R}^{{n}}$ is the vector field, {and} 
$\mathcal{X}$ is the entire state space.
A function $g_d: {\mathcal{X}\times \mathcal{X}} \to \mathbb{R}^{n}$ is a {discrete-time mixed-monotone} decomposition mapping for the vector field $g$ if it satisfies the following conditions: i) $g_d({x,x})=g({x})$, ii) ${\color{black}g_d}$ is monotone increasing in its first argument, i.e., $\hat{{x}}\ge {x} \Rightarrow g_d(\hat{{x}},{x}') \geq g_d(x,x')$, and iii) $g_d$ is monotone decreasing in its second argument, i.e., $\hat{x}\ge x \Rightarrow g_d({x}',\hat{{x}}) \leq g_d({x}',{x})$.
\end{defn}

\begin{prop}[Tight and Tractable Decomposition Functions for JSS Mappings] \label{prop:tight_decomp}\cite[Proposition 4 \& Lemma 3]{moh2022intervalACC}
Suppose $\mu: {\mathcal{Z}} \subset \mathbb{R}^{n_z} \to \mathbb{R}^p$ is a JSS mapping on its domain. Then, for each $\mu_i$, $i \in \mathbb{N}_p$, its tight decomposition function is: 
\begin{align}\label{eq:JJ_decomp}
\mu_{d,i}({z}_1,{z}_2)\hspace{-.0cm}=\hspace{-.0cm}\mu_i(D^i{z}_1\hspace{-.0cm}+\hspace{-.0cm}(I_{n}\hspace{-.0cm}-\hspace{-.0cm}D^i){z}_2), 
\end{align}
for any ordered ${z_1, z_2 \in \mathcal{Z}}$, with a binary diagonal matrix $D^i$ that is determined by 
the vertex of the interval $[{z}_1,{z}_2]$ that minimizes the function $\mu_i$ (if $z_1 < z_2$) or the vertex of the interval $[z_2,z_1]$ that maximizes $\mu_i$ (if $z_2 \leq z_1$), 
i.e., 
$$D^i=\textstyle{\mathrm{diag}}(\max(\textstyle{\mathrm{sgn}}({\overline{J}^{\mu}_i}),\mathbf{0}_{1,{n_z}})).$$
Moreover, if the JSS mapping $\mu$ is a remainder term of a JSS decomposition of a function $f$ as discussed in Proposition \ref{prop:JSS_decomp}, then for any interval domain $\underline{z} \leq z \leq \overline{z}$, with $z,\underline{z},\overline{z} \in \mathcal{Z}$ and $\varepsilon \triangleq \overline{z}-\underline{z}$, the following inequality holds:
$\delta^{\mu}_d \sy{\triangleq \mu_d(\overline{z},\underline{z})-\mu_d(\underline{z},\overline{z})} \leq \overline{F}_{\mu}\varepsilon$, 
\yo{with $\overline{F}_{\mu}\hspace{-.cm}\triangleq\hspace{-.cm}2\max(\overline{J}_f\hspace{-.cm}-\hspace{-.cm}H,\mathbf{0}_{p,n_z})\hspace{-.07cm}-\hspace{-.07cm}\underline{J}_f\hspace{-.07cm}+\hspace{-.07cm}H$}
and $H \in \mathbb{R}^{p \times n_z}$ given in Proposition \ref{prop:JSS_decomp}.
\end{prop}
Consequently, by applying Proposition \ref{prop:tight_decomp} to the Jacobian sign-stable decomposition 
obtained using Proposition \ref{prop:JSS_decomp}, a tight and tractable decomposition function can be obtained (cf. details 
in \cite{moh2022intervalACC}). Furthermore, in 
\yo{the case that} the mapping is not JSS, a tractable algorithm has been introduced in \cite[Algorithm 1]{khajenejad2021tight} to compute \emph{tight remainder-form decomposition functions} for a very broad class of nonlinear functions. 
\begin{defn}[Embedding System]\cite[Definition 6]{khajenejad2022h}\label{def:embedding}
For a discrete-time dynamical system $x_{k+1}= g(x_k)$ defined over mapping  ${g}:{\mathcal{X}} \subset \mathbb{R}^{n} \to \mathbb{R}^{{n}}$ with a corresponding decomposition function ${g}_{d}(\cdot)$, its embedding system is a $2n$-dimensional system with initial condition $\begin{bmatrix} \overline{x}_0^\top & \underline{x}_0^\top\end{bmatrix}^\top$  defined as
$\begin{bmatrix}{\underline{x}}^\top_{k+1} & {\overline{x}}^\top_{k+1} \end{bmatrix}^\top=\begin{bmatrix}  \underline{g}_{d}^\top(\underline{x}_k , \overline{x}_k) &  \overline{g}_{d}^\top(\overline{x}_k , \underline{x}_k) \end{bmatrix}^\top$.
\end{defn}
Note that according to \cite[Proposition 3]{khajenejad2021tight}, the embedding system in Definition \ref{def:embedding} with decomposition function $g_d$ corresponding to the dynamics $x_{k+1}=g(x_{k})$ has a \emph{state framer property}, i.e., its solution is guaranteed to frame the unknown state trajectory $x_k$, i.e.,  $\underline{x}_k\le x_k\le \overline{x}_k$ for all $k\in\mathbb{N}$.

\moha{Next, we will briefly restate our previous result in \cite{singh2018mesh}, tailoring it \yongs{specifically for intervals to help with computing} 
affine bounding functions for our functions.}
\begin{prop}\cite[Affine Outer-Approximation]{singh2018mesh}\label{prop:affine abstractions}
Consider the 
\yo{function} $g(.):\mathcal{B} \subset \mathbb{R}^n \to \mathbb{R}^m$, where $\mathcal{B}$ is an interval 
with $\overline{x},\underline{x},\mathcal{V}_{\mathcal{B}}$ being its maximal, minimal and set of vertices, respectively. Suppose $\overline{A}_{\mathcal{B}},\underline{A}_{\mathcal{B}}, \overline{e}_{\mathcal{B}}, \underline{e}_{\mathcal{B}},\theta_{\mathcal{B}}$ \moham{is a solution of} the following \yongs{linear program (LP):}   
\begin{align} \label{eq:abstraction}
&\min_{\theta,\overline{A},\underline{A},\overline{e},\underline{e}} {\theta} \\
\nonumber & \quad \quad s.t \ \underline{A} {x}_{s}+\underline{e}+\sigma \leq g({x}_{s}) \leq \overline{A} {x}_{s}+\overline{e}-\sigma, \\
\nonumber &\quad \quad \quad \ (\overline{A}-\underline{A}) {x}_{s}+\overline{e}-\underline{e}-2\sigma \leq \theta \mathbf{1}_m , \ \forall x_s \in \mathcal{V}_{\mathcal{B}},
\end{align}  
where $\mathbf{1}_m \in \mathbb{R}^m$ is a vector of ones and $\sigma$ can be computed via \yongs{\cite[Proposition 1]{singh2018mesh}} 
\moha{for different function classes}. Then, $\underline{A}_{\mk{\mathcal{B}}} {x}+\underline{e}_{\mk{\mathcal{B}}} \leq g(x) \leq  \overline{A}_{\mk{\mathcal{B}}} {x}+\overline{e}_{\mk{\mathcal{B}}}, \forall x \in \mathcal{B}$.  
\end{prop}
\begin{cor} \label{cor:parallel_abst}
By taking the average of upper and lower affine abstractions and adding/subtracting half of the maximum distance, it is straightforward to \yo{``parallelize''} the above upper and lower abstractions as ${A}_g{x}+\underline{\epsilon} \leq g(x) \leq  A_g{x}+\overline{\epsilon}$, or equivalently
$g(x)=A_gx+\epsilon, \epsilon \in [\underline{\epsilon},\overline{\epsilon}]$,
where $A_g\triangleq (1/2)(\overline{A}+\underline{A})$, $\underline{\epsilon}\triangleq (1/2)(\overline{e}+\underline{e}-\theta \mathbf{1}_m)$ and $\overline{\epsilon} \triangleq (1/2)(\overline{e}+\underline{e}+\theta \mathbf{1}_m)$. We call $A_g$ and $\epsilon$ \yo{the} parallel affine outer-approximation slope and outer-approximation error of function $g$ on 
$\mathcal{B}$, respectively. 
\end{cor}
\section{Problem Formulation} \label{sec:Problem}
\noindent\textbf{\emph{System Assumptions.}} 
Consider the nonlinear discrete-time system with unknown inputs and bounded noise 
\begin{align} \label{eq:system}
\begin{array}{rl}
x_{k+1}&=f(x_k)+Ww_k+G d_k ,\\
y_k&=h(x_k) +Vv_k + H d_k  , 
\end{array}
\end{align}
where at time \yo{$k \in \mathbb{N}$,} 
$x_k \in \mathcal{X} \subset \mathbb{R}^n$, $d_k \in \mathbb{R}^p$ and $y_k \in \mathbb{R}^l$ are the state vector, 
unknown input vector, and  measurement vector, respectively. The process and measurement noise signals $w_k \in \mathbb{R}^n$ and $v_k \in \mathbb{R}^l$ are assumed to be bounded, \yo{i.e., $w_k \in \mathcal{W} \triangleq [\underline{w},\overline{w}\}$, $v_k \in \mathcal{V} \triangleq [\underline{v},\overline{v}]$ with} 
known lower and upper bounds, $\underline{w}$, $\overline{w}$ and $\underline{v}$, $\overline{v}$, respectively. 
We also assume \yongn{that} lower and upper bounds for the initial state, $\underline{x}_0$ and $\overline{x}_0$, {are} available, i.e., $\underline{x}_0 \leq x_0 \leq \overline{x}_0$. The functions $f:\mathbb{R}^n \rightarrow \mathbb{R}^n$, $h:\mathbb{R}^n \rightarrow \mathbb{R}^l$ and matrices $W$, $V$, $G$ and $H$ are known 
and of appropriate dimensions, where $G$ and $H$ 
encode the \emph{locations} at which the unknown input \yongn{(or attack)} signal can affect the system dynamics and measurements. Note that no assumption is made on $H$ to be either the zero matrix (no direct feedthrough), or to have full column rank when there is direct feedthrough (in contrast \yongs{to} \yo{\cite{khajenejad2020simultaneousfullrank})}.  


\noindent \textbf{\emph{Unknown Input (or Attack) Signal Assumptions.}} 
The unknown inputs $d_k$ \yo{(representing false data injection attack signals)} are not constrained to follow any model nor to be a signal of any type (random or strategic), hence no prior `useful' knowledge of the dynamics of $d_k$ is available (independent of $\{d_\ell\}$ $\forall k\neq \ell$, $\{w_\ell\}$ and $\{v_\ell\}$ $ \forall  \ell$). We also do not assume that $d_k$ is bounded or has known bounds and thus, $d_k$ is suitable for representing adversarial 
attack signals.

Next, we briefly introduce a similar system transformation as in \cite{yong2018simultaneous}, which will be used later in our observer structure.

 \noindent \textbf{\emph{System Transformation.}} Let $p_{H}\triangleq {\rm rk} (H)$. Similar to \cite{yong2018simultaneous}, by applying singular value decomposition, we have $H= \begin{bmatrix}U_{1}& U_{2} \end{bmatrix} \begin{bmatrix} \Xi & 0 \\ 0 & 0 \end{bmatrix} \begin{bmatrix} E_{1}^{\, \top} \\ E_{2}^{\, \top} \end{bmatrix}$ with $E_{1} \in \mathbb{R}^{p \times p_{H}}$, $E_{2} \in \mathbb{R}^{p \times (p-p_{H})}$, $\Xi \in \mathbb{R}^{p_{H} \times p_{H}}$ (a diagonal matrix of full rank; so we can define $S\triangleq \Xi^{-1}$),
 $U_{1} \in \mathbb{R}^{l \times p_{H}}$ and $U_{2} \in \mathbb{R}^{l \times (l-p_{H})}$. 
Then, since 
$D\triangleq \begin{bmatrix} E_{1} & E_{2} \end{bmatrix}$ is unitary: 
\begin{align}
d_k =E_{1} d_{1,k}+E_{2} d_{2,k}, \ d_{1,k}=E_{1}^\top d_k, \ d_{2,k}=E_{2}^\top d_k.  \label{eq:d12}
\end{align}
 Finally, by defining $T_1\triangleq U^\top_1,T_2 \triangleq U^\top_2$, the output equation can be decoupled, by which system \eqref{eq:system} can be rewritten as: 
 \begin{align}
\label{eq:stateq}
\begin{array}{rl}x_{k+1}&=f(x_k)+Ww_k+G_1 
 d_{1,k}+G_2d_{2,k},\\
 z_{1,k}&= h_1(x_k) +  V_1v_{k}+\Xi d_{1,k}, 
 \\
  z_{2,k}&= \yo{h_2(x_k) + V_2v_{k}}, 
  \end{array}
  \end{align}
  \yo{where $h_i(x_k)=T_i h(x_k)$, $\forall i \in \{1,2\}$} 
  and $K_i \triangleq T_iK_i, \forall K \in \{V,G\},\forall i \in \{1,2\}$. 
 
Moreover, \yongs{we assume the following, which is satisfied for a broad range of nonlinear \yongn{functions \cite{yang2019tight}}: }

\vspace{-0.1cm}
\begin{assumption}\label{assumption:mix-lip}
Functions $f,h$ have bounded Jacobians over the state space $\mathcal{X}$ with known/computable Jacobian bounds.
\begin{assumption}\label{ass:rank}
\yo{The JSS decomposition of $h_2(x_k)$ via Proposition \ref{prop:JSS_decomp} given by $h_2(x_k)=C_2 x_k+\psi_2(x_k)$ is such that $\psi_2$ is JSS and further,}
$C_2G_2$ has full column rank\footnote{\mk{In the special case that $G=0$, we would require $G_2$ to be empty (and this does happen when $H$ has full rank), in which case $C_2G_2$ being full rank is satisfied by assumption.}}. Consequently, there exists $M_2 \triangleq (C_2G_2)^\dagger$ such that $M_2C_2G_2=I$. 
\end{assumption}
\begin{assumption} \label{ass:abs_inv}
\yo{(Only needed \mk{when the observations are nonlinear, i.e.,}} if $\psi_2(x_k)\neq 0$) \mk{The entire state space $\mathcal{X} \subset \mathbb{R}^n$ is bounded. Moreover,} $A_g$ is invertible, where $A_g \in \mathbb{R}^{n \times n}$ is the parallel affine outer-approximation slope (cf. Proposition \ref{prop:affine abstractions} and Corollary \ref{cor:parallel_abst}) of the function $g(x) \triangleq x+ 
\yo{G_2}M_2\psi_2(x)$ over the entire state space. 
\end{assumption}
\end{assumption} \vspace{-0.1cm}
Further, we formally define the notions of \emph{framers}, \emph{correctness} and \emph{stability} that are used throughout the paper. 
\begin{defn}[Interval {Framers}]\label{defn:framers}
Given the nonlinear plant \eqref{eq:system} (equivalently \eqref{eq:stateq}), 
the sequences $\{\overline{x}_k,\underline{x}_k\}_{k=0}^{\infty}\subset \mathbb{R}^n$ and $\{\overline{d}_k,\underline{d}_k\}_{k=0}^{\infty}\subset \mathbb{R}^p$ are called upper and lower framers for the states and inputs of the system in \eqref{eq:system}, respectively, if 
\begin{align*}
\forall k \in \yong{\mathbb{N}}, {\forall w_k \in \mathcal{W},\forall v_k \in \mathcal{V}}, \ \underline{\nu}_k \leq \nu_k \leq \overline{\nu}_k, \forall \nu \in \{x,d\}.
\end{align*}
In other words, starting from the initial interval $\underline{x}_0 \leq x_0 \leq \overline{x}_0$, the true state of the system in \eqref{eq:system}, $x_k$, and the unknown input $d_k$ are guaranteed to evolve within the interval flow-pipe $[\underline{x}_k,\overline{x}_k]$ and bounded within the interval $[\underline{d}_k,\overline{d}_k]$, for all $(k,{w_k,v_k)} \in {\mathbb{N} \times \mathcal{W} \times \mathcal{V}}$, respectively. Finally, 
any dynamical system (i.e., tractable algorithm) that returns upper and lower framers for the states and unknown inputs of system  \ref{eq:system} is called a resilient interval {framer} for \eqref{eq:system}. 
\end{defn}
\begin{defn}[{Framer} Error]\label{defn:error}
Given 
state and input framers $\{\underline{x}_k \leq \overline{x}_k\}_{k=0}^{\infty}$ and $\{\underline{d}_k \leq \overline{d}_k\}_{k=1}^{\infty}$, the sequences $\{e^x_k \triangleq \overline{x}_k-\underline{x}_k\}_{k=0}^{\infty}$ and $\{e^d_k \triangleq \overline{d}_k-\underline{d}_k\}_{k=1}^{\infty}$ are called
 the state and input {framer} errors, respectively. It 
 easily \mk{follows from}
Definition \ref{defn:framers} that 
$e^{\nu}_k \geq 0, \forall k \in {\mathbb{N}}, \forall \nu \in \{x,d\}.$  
\end{defn}
\begin{defn}[{Input-to-State} Stability and {Interval Observer}]\label{defn:stability}
An interval {framer} is {input-to-state} stable {(ISS)}, if the {framer} state error (cf. Definition \ref{defn:error}) {is bounded as follows: 
\begin{align}
\forall k \in \mathbb{N}, \ \|e^x_k\|_{2} \leq \beta(\|e^x_0\|_{2},k)+\alpha(\|\delta\|_{\ell_{\infty}}),
\end{align} 
where {$\delta \triangleq [(\delta^w)^\top \ (\delta^v)^\top]^\top \triangleq [(\overline{w}-\underline{w})^\top \ (\overline{v}-\underline{v})^\top]^\top$},
$\beta$ and $\alpha$ are functions of classes\footnote{{A function $\alpha: \mathbb{R}_{\geq 0} \to \mathbb{R}_{\geq 0}$ is of class $\mathcal{K}$ if it is continuous, positive definite, and strictly increasing and is of class $\mathcal{K}_{\infty}$ if it is also unbounded. Moreover, $\lambda : \mathbb{R}_{\geq 0} \to \mathbb{R}_{\geq 0}$ is of class $\mathcal{KL}$ if for each fixed $t\geq 0$, $\lambda(\cdot,t)$ is of class $\mathcal{K}$ and for each fixed $s \geq 0$, $\lambda(s,t)$ decreases to zero as $t \to \infty.$}} $\mathcal{KL}$ and $\mathcal{K}_{\infty}$, respectively, {and $\|\delta\|_{\ell_{\infty}} \triangleq \sup_{k \in \sy{\mathbb{N}}}
\|\delta_k\|_2=\|\delta\|_2$ is the $\ell_{\infty}$ signal norm}.}
{An ISS \sy{resilient} interval framer is called \sy{a resilient} interval observer.}
\end{defn}
\begin{defn}[$\mathcal{H}_{\infty}$-Optimal Resilient Interval Observer]\label{defn:H_inf}
A resilient interval framer design 
is $\mathcal{H}_{\infty}$-optimal 
if the \yo{$\mathcal{H}_{\infty}$} gain of the framer error system $\tilde{\mathcal{G}}$, i.e., $\|\tilde{\mathcal{G}}\|_{\yo{\mathcal{H}_{\infty}}}$ is minimized, where 
$\|\tilde{\mathcal{G}}\|_{\yo{\mathcal{H}_{\infty}}} \triangleq \sup {\{}\frac{\|e^x\|_{\ell_2}}{\|\delta\|_{\ell_2}},\delta \ne 0{\}}$, {and} 
$\|s\|_{\ell_2} \triangleq \sqrt{\sum_{0}^{\infty}\|s_k\|_2^2}$ is the $\ell_2$ signal norm for $s \in \{e^x,\delta\}$.
\end{defn}
\mk{Using the above, we aim to address the following problem.}
\begin{problem}\label{prob:SISIO}
Given the nonlinear system in \eqref{eq:system}, as well as Assumptions {\ref{assumption:mix-lip}--\ref{ass:abs_inv}}, 
synthesize {an ISS and {$\mathcal{H}_{\infty}$-optimal} resilient interval observer (cf. Definitions \ref{defn:framers}--\ref{defn:H_inf})}.
\end{problem}

\section{Resilient Interval Observer Design} \label{sec:observer}
In this section, we describe 
the proposed resilient interval observer as well as analyze its correctness and ISS properties.
\subsection{Interval Framer Design} \label{sec:obsv}
Our strategy for designing resilient interval observers in the presence of unknown inputs has three steps. First, we obtain an equivalent representation of the system in \eqref{eq:system} by introducing some auxiliary state variables, such that the equivalent system is not affected by the attack signal. Then, inspired by our previous work on synthesizing interval observers for nonlinear systems \cite{moh2022intervalACC,khajenejad2022h} we will design embedding systems (cf. Definition \ref{def:embedding}) for the equivalent system representation, which returns state framers. Finally, we obtain input framers \yo{(with a one-step delay \sy{since $d_{2,k}$ does not appear in the measurements $z_{1,k}$ and $z_{2,k}$ in \eqref{eq:stateq}})} as functions of the computed state framers.

\yo{First,} note that from \eqref{eq:stateq} 
and with $S \triangleq \Xi^{-1}$, $d_{1,k}$ can be computed as a function of the state at current time as follows:
\begin{align}\label{eq:d1}
d_{1,k}=S(z_{1,k}-h_1(x_k)-V_1v_k).
\end{align}
Next, we introduce an auxiliary state variable as:
\begin{align}\label{eq:aux}
\xi_k \triangleq\hspace{-.05cm} x_k\hspace{-.05cm}-\hspace{-.05cm}N(z_{2,k}-V_2v_k-\psi_2(x_k))=\hspace{-.cm}(I\hspace{-0.05cm}-\hspace{-0.05cm}NC_2)x_k, 
\end{align}
where \mk{the equality follows from \eqref{eq:stateq} and Assumption \ref{ass:rank}. Moreover,} $N \in \mathbb{R}^{n \times (l-\tilde{p})}$ is a to-be-designed gain to cancel out the effect of the unknown input in the state equation. This is done through the following lemma.\vspace{-0.1cm}
\begin{lem}
Suppose Assumption \ref{ass:rank} holds and let \mk{$N=G_2M_2=G_2(C_2G_2)^\dagger$} and $S \triangleq \Xi^{-1}$. Then,  the value of the auxiliary state $\xi_k$ at \sy{time step $k+1$} 
can be computed as:

\vspace{-.3cm}
{\small
\begin{align}\label{eq:sys_equivalent_3}
\hspace{-.2cm}\xi_{k+1}&\hspace{-.1cm}=\hspace{-.1cm}(I\hspace{-.1cm}-\hspace{-.1cm}NC_2)(f(x_k)\hspace{-.1cm}+\hspace{-.1cm}G_1S(z_{1,k}\hspace{-.1cm}-\hspace{-.1cm}h_1(x_k)\hspace{-.1cm}-\hspace{-.1cm}V_1v_k)\hspace{-.1cm}+\hspace{-.1cm}Ww_k).
\end{align}}\vspace{-0.45cm}
\end{lem}
\begin{proof}
By plugging $d_{1,k}$ from \eqref{eq:d1} into \eqref{eq:stateq}, we obtain

\vspace{-.3cm}
{\small
\begin{align}\label{eq:x_d_2}
x_{k+1}\hspace{-.1cm}=\hspace{-.1cm}f(x_k)\hspace{-.1cm}+\hspace{-.1cm}G_1S(z_{1,k}\hspace{-.07cm}-\hspace{-0.07cm}h_1(x_k)\hspace{-0.07cm}-\hspace{-.07cm}V_1v_k)\hspace{-.1cm}+\hspace{-.1cm}Ww_k\hspace{-.1cm}+\hspace{-.1cm}G_2d_{2,k}.
\end{align}
}\vspace{-0.4cm}

\noindent This, together with the second equality in \eqref{eq:aux} and \sy{the above choice of $N$ such that} 
$(I\hspace{-.1cm}-\hspace{-.1cm}NC_2)G_2=0$, 
returns \eqref{eq:sys_equivalent_3}.
\end{proof}\vspace{-0.2cm}

The evolution of the auxiliary state $\xi_k$ in \eqref{eq:sys_equivalent_3} is independent of the unknown input and hence, we can compute propagated framers for $\xi_k$ leveraging embedding systems \yo{(cf. Proposition \ref{def:embedding})}. However, we do not have a way of \yo{directly} retrieving the propagated framers for the original states, i.e., $\{\underline{x}_k,\overline{x}_k\}$ in terms of $\{\underline{\xi}_k,\overline{\xi}_k\}$ \yo{from the second equality of \eqref{eq:aux}, since $I-NC_2=I-G_2(C_2 G_2)^\dagger C_2$ can be shown to be not invertible}. To overcome this difficulty, given Assumption \ref{ass:abs_inv}, 
we introduce a new auxiliary state:
\begin{align}\label{eq:gamma}
\gamma_k \triangleq x_k-\Lambda(N(z_{2,k}-V_2v_k)-\epsilon_k),
\end{align}
{\color{black}with} $\Lambda \triangleq A_g^{-1}$, {\color{black}where} $A_g$ and $\epsilon_k \in [\mk{\underline{\epsilon}},\overline{\epsilon}]$  are parallel affine outer-approximation slope and approximation error of the mapping $g(x) \triangleq x+ 
\yo{G_2}M_2\psi_2(x)$ on the entire space $\mathcal{X}$ (cf. Proposition \ref{prop:affine abstractions}, Corollary \ref{cor:parallel_abst} and Assumption \ref{ass:abs_inv}).\vspace{-0.1cm}
\begin{prop}\label{prop:two_aux}
Given Assumption \ref{ass:abs_inv}, the two auxiliary states $\gamma_k$ and $\xi_k$ are linearly related as: 
\begin{align}\label{eq:two_aux}
\gamma_k=\Lambda \xi_k.
\end{align}
\end{prop}
\begin{proof}
Computing parallel affine outer-approximation of the mapping 
\yo{$g(x_k)=A_g x_k+\epsilon_k$ and applying \eqref{eq:aux}, we obtain} 
\begin{align*}
g(x_k)&\triangleq \yo{x_k}+N\psi_2(x)=\xi_k+N(z_{2,k}-V_2v_k) \\ 
\yo{\Rightarrow} A_gx_k&=\xi_k+N(z_{2,k}-V_2v_k)-\epsilon_k, \quad \epsilon_k \in [\underline{\epsilon},\overline{\epsilon}],
\end{align*}
\yo{from} which \mo{and} given Assumption \ref{ass:abs_inv} (\yo{that $A_g$ is invertible, with $\Lambda=A_g^{-1}$), we have}  
\begin{align}\label{eq:x_xi}
x_k=\Lambda(\xi_k+N(z_{2,k}-V_2v_k)-\epsilon_k), \quad \epsilon_k \in [\underline{\epsilon},\overline{\epsilon}].
\end{align}
Plugging $x_k$ from \eqref{eq:x_xi} into \eqref{eq:gamma} returns the results.
\end{proof}
 \mo{We are \yo{now} ready to propose an input and state resilient interval framer, i.e., the following discrete-time dynamical system \eqref{eq:framers}--\eqref{eq:d_framers}, which by construction, {outputs}\yo{/returns} framers for the original states $\{x_k\}_{k=0}^{\infty}$ and the unknown input signal $\{d_k\}_{k=1}^{\infty}$ of system \eqref{eq:system}. The details of \yo{the framer construction/design} 
 will be provided in the proof of Theorem \ref{thm:state_framer}. The proposed resilient interval framer is as follows:}
\begin{align}\label{eq:framers}
\begin{array}{rlllll}
\hspace{-.2cm}\underline{\gamma}_{k+1}&\hspace{-.2cm}=\hspace{-.1cm}(A\hspace{-.1cm}-\hspace{-.1cm}LC_2)^\oplus \underline{\gamma}_k\hspace{-.1cm}-\hspace{-.1cm}(A\hspace{-.1cm}-\hspace{-.1cm}LC_2)^\ominus \overline{\gamma}_k\hspace{-.1cm}+\hspace{-.1cm}\rho_d(\underline{x}_k,\overline{x}_k)\\
&\hspace{-.2cm}+D^{\ominus} \underline{\epsilon}\hspace{-.1cm}-\hspace{-.1cm}D^\oplus \overline{\epsilon}\hspace{-.1cm}+\hspace{-.1cm}L^{\ominus} \psi_{2,d}(\underline{x}_k,\overline{x}_k)\hspace{-.1cm}-\hspace{-.1cm}L^{\oplus} \psi_{2,d}(\overline{x}_k,\underline{x}_k)\\
&\hspace{-.2cm}+\hat{V}^{\ominus} \underline{v}-\hat{V}^\oplus \overline{v}+\hat{W}^{\ominus} \underline{w}-\hat{W}^\ominus \overline{w}+\hat{z}_k,\\
\hspace{-.2cm}\overline{\gamma}_{k+1}&\hspace{-.2cm}=\hspace{-.1cm}(A\hspace{-.1cm}-\hspace{-.1cm}LC_2)^\oplus \overline{\gamma}_k\hspace{-.1cm}-\hspace{-.1cm}(A\hspace{-.1cm}-\hspace{-.1cm}LC_2)^\ominus \underline{\gamma}_k\hspace{-.1cm}+\hspace{-.1cm}\rho_d(\overline{x}_k,\underline{x}_k)\\
&\hspace{-.2cm}+D^{\ominus} \overline{\epsilon}\hspace{-.1cm}-\hspace{-.1cm}D^\oplus \underline{\epsilon}\hspace{-.1cm}+\hspace{-.1cm}L^{\ominus} \psi_{2,d}(\overline{x}_k,\underline{x}_k)\hspace{-.1cm}-\hspace{-.1cm}L^{\oplus} \psi_{2,d}(\underline{x}_k,\overline{x}_k)\\
&\hspace{-.2cm}+\hat{V}^{\ominus} \overline{v}-\hat{V}^\oplus \underline{v}+\hat{W}^{\oplus} \overline{w}-\hat{W}^\ominus \underline{w}\hspace{-.1cm}+\hspace{-.1cm}\hat{z}_k,
\end{array}
\end{align}

\vspace{-.15cm}
{\small
\begin{align}\label{eq:framers_2}
\begin{array}{rl}
\hspace{-.4cm}\underline{x}_{k}&\hspace{-.2cm}=\hspace{-.05cm}\underline{\gamma}_k\hspace{-.1cm}+\hspace{-.1cm}\Lambda Nz_{2,k}\hspace{-.1cm}+\hspace{-.1cm}\Lambda^{\ominus}\underline{\epsilon}\hspace{-.1cm}-\hspace{-.1cm}\Lambda^{\oplus}\overline{\epsilon}\hspace{-.1cm}+\hspace{-.1cm}(\Lambda NV_2)^{\ominus}\underline{v}\hspace{-.1cm}-\hspace{-.1cm}(\Lambda NV_2)^{\oplus}\overline{v},\\
\hspace{-.4cm}\overline{x}_k &\hspace{-.2cm}= \hspace{-.05cm}\overline{\gamma}_k\hspace{-.1cm}+\hspace{-.1cm}\Lambda Nz_{2,k}\hspace{-.1cm}+\hspace{-.1cm}\Lambda^{\ominus}\overline{\epsilon}\hspace{-.1cm}-\hspace{-.1cm}
\Lambda^{\oplus}\underline{\epsilon}\hspace{-.1cm}+\hspace{-.1cm}(\Lambda NV_2)^{\ominus}\overline{v}\hspace{-.1cm}-\hspace{-.1cm}(\Lambda NV_2)^{\oplus}\underline{v},
\end{array}
\end{align}
}
\vspace{-.15cm}
\begin{align}\label{eq:d_framers}
\begin{array}{rlll}
\hspace{-.2cm}\underline{d}_{k-1}&\hspace{-.2cm}=\Phi^\oplus\underline{x}_k-\Phi^\ominus\overline{x}_k+\mj{\kappa_d}(\underline{x}_{k-1},\overline{x}_{k-1})\hspace{-0.05cm}+\hspace{-0.05cm}A_zz_{1,k-1}\\
&\hspace{-.2cm}+A_v^\oplus\underline{v}-A_v^{\ominus}\overline{v}+\Phi^{\ominus}\underline{w}-\Phi^{\oplus}\overline{w},\\
\hspace{-.2cm}\overline{d}_{k-1}&\hspace{-.2cm}=\Phi^\oplus\overline{x}_k-\Phi^\ominus\underline{x}_k+\mj{\kappa_d}(\overline{x}_{k-1},\underline{x}_{k-1})\hspace{-0.05cm}+\hspace{-0.05cm}A_zz_{1,k-1}\\
&\hspace{-.2cm}+A_v^\oplus\overline{v}-A_v^{\ominus}\underline{v}+\Phi^{\ominus}\overline{w}-\Phi^{\oplus}\underline{w},
\end{array}
\end{align}
where $S \triangleq \Xi^{-1}$, $N=G_2M_2$ and $\Lambda \triangleq A^{-1}_g$. Furthermore, $L \in \mathbb{R}^{n \times (l-p_{H})}$ is an arbitrary matrix (observer gain) which will be designed later in Theorem \ref{thm:stability} to \sy{yield} 
stability and optimality of the proposed framers. Moreover, $A \in \mathbb{R}^{n \times n}$ and $\rho:\mathcal{X} \subset \mathbb{R}^n \to \mathbb{R}^n$ are obtained by applying JSS decompositions (cf. Proposition \ref{prop:JSS_decomp}) on the mapping $ \tilde{f}(x) \triangleq \Lambda (I-NC_2)(f(x)-G_1Sh_1(x))$, while $\psi_{2,d}$ and $\rho_d$ are tight decomposition functions for the JSS mappings $\rho$ and $\psi_2$, respectively, computed through Proposition \ref{prop:JSS_decomp}. Further, 

\vspace{-.3cm}
{\small
\begin{align}\label{eq:D}
\begin{array}{rlll}
\hspace{-.2cm}\hat{V} &\hspace{-.2cm}\triangleq 
(A-LC_2)\Lambda NV_2+LV_2+\Lambda(I-NC_2)G_1SV_1,  \\
\hspace{-.2cm}\Phi &\hspace{-.2cm}\triangleq E_2M_2C_2, A_v \triangleq (\Phi G_1-E_1)SV_1, A_z \triangleq (E_1\hspace{-.1cm}-\hspace{-.1cm}\Phi G_1)S,  \\
\hspace{-.2cm}D &\hspace{-.2cm}\triangleq (A-LC_2)\Lambda, \hat{W} \triangleq \Lambda(I-NC_2)W, \\ 
\hspace{-.2cm}\hat{z}_k &\hspace{-.2cm}\triangleq \Lambda(I-NC_2)G_1Sz_{1,k}+(L+(A-LC_2)\Lambda N)z_{2,k}, 
\end{array}
\end{align}
}
\yo{and} $\mj{\kappa_d}$ is the decomposition function of the mapping $\mj{\kappa}(x) \triangleq (\Phi G_1-E_1)Sh_1(x)-\Phi f(x)$, computed via \cite[Theorem 1]{yang2019sufficient}. 
Finally, $A_g$ and $\epsilon_k \in [\underline{\epsilon},\overline{\epsilon}]$ are computed 
\yo{via} Corollary \ref{cor:parallel_abst}.

The following theorem formalizes the state and input framer/correctness property of the proposed resilient interval observer \eqref{eq:framers}--\eqref{eq:d_framers} with respect to the original system \eqref{eq:system}.
\begin{thm}\label{thm:state_framer}
Suppose Assumptions \ref{assumption:mix-lip}--\ref{ass:abs_inv} hold. 
 Then, the sequences $\{\underline{x}_k,\overline{x}_k\}_{k=0}^{\infty}$ and $\{\underline{d}_k,\overline{x}_k\}_{k=1}^{\infty}$ obtained from the system \eqref{eq:framers}--\eqref{eq:d_framers}, construct framers for the states and unknown input signal of \eqref{eq:system}, respectively, i.e., $\underline{\nu}_k \leq \nu_k \leq \overline{\nu}_k, \forall \nu \in \{x,d\}, \forall k \in \mathbb{N},\forall w_k \in \mathcal{W},\forall v_k \in \mathcal{V}$.
\end{thm}
\begin{proof}
\yo{From \eqref{eq:sys_equivalent_3} and  \eqref{eq:two_aux}, we obtain} 
\begin{align*}
\gamma^+_k\hspace{-.1cm}=\hspace{-.1cm}\Lambda(I\hspace{-.1cm}-\hspace{-.1cm}NC_2)(f(x_x)\hspace{-.1cm}+\hspace{-.1cm}G_1S(z_{1,k}\hspace{-.1cm}-\hspace{-.1cm}h_1(x_k)\hspace{-.1cm}-\hspace{-.1cm}V_1v_k)\hspace{-.1cm}+\hspace{-.1cm}Ww_k).
\end{align*}
Adding the \emph{zero term} $L(z_{2,k}-C_2x_k-\psi_2(x_k)-V_2v_k)$ to the right hand side of \mk{the above equation} and applying mixed-monotone decompositions \mo{on the mapping $ \tilde{f}(x) \triangleq \Lambda (I-NC_2)(f(x)-G_1Sh_1(x))$ to decompose it as $\tilde{f}(x)=Ax+\rho(x)$ (cf. Proposition \ref{prop:JSS_decomp} for more details),} yields:
\begin{align}\label{eq:gamma_dyn_2}
\hspace{-.2cm}\gamma_{k+1}\hspace{-.1cm}=\hspace{-.1cm}(A\hspace{-.1cm}-\hspace{-.1cm}LC_2)x_k\hspace{-.1cm}+\hspace{-.1cm}\rho(x_k)\hspace{-.1cm}-\hspace{-.1cm}L\psi_2(x_k)\hspace{-.1cm}-\hspace{-.1cm}\tilde{V}v_k\hspace{-.1cm}+\hspace{-.1cm}\tilde{W}w_k\hspace{-.1cm}+\hspace{-.1cm}\tilde{z}_k,
\end{align}
where $\tilde{V} \triangleq \Lambda (I-NC_2)G_1SV_1+LV_2$,
 $\tilde{W} \triangleq \Lambda(I-NC_2)W$ and
  $\tilde{z}_k \triangleq \Lambda(I-NC_2)G_1Sz_{1,k}+Lz_{2,k}$. Then, by computing $x_k$ in terms of $\gamma_k$ from \eqref{eq:gamma} and plugging it back into the linear terms in the right-hand side of \eqref{eq:gamma_dyn_2}, we obtain 
  \begin{align}\label{eq:sys_equivalent_4}\small
\begin{array}{rl}
\gamma_{k+1} 
&=(A-LC_2)\gamma_k+\rho(x_k)-L\psi_2(x_k)\\
&-\hat{V}v_k+\hat{W}w_t-D\epsilon_k+\hat{z}_k,
\end{array}\normalsize
\end{align}\vspace{-0.35cm}

 \noindent with $\hat{V},D,\hat{W}$ and $\hat{z}_k$ given in \eqref{eq:D}. Next, by applying Proposition \ref{prop:tight_decomp} and \cite[Lemma 1]{efimov2013interval}, we construct the embedding system \eqref{eq:framers} for \eqref{eq:sys_equivalent_4}, which implies $\underline{\gamma}_k \leq \gamma_k \leq\overline{\gamma}_k, \forall k \in \mathbb{N}$, by construction. Further, the results in \eqref{eq:framers_2} follow from applying \cite[Lemma 1]{efimov2013interval} on \eqref{eq:gamma} to compute framers of $x_k$ in terms of the framers of $\gamma_k$. 
 
 To obtain input framers, note that multiplying both sides of \eqref{eq:x_d_2} by $M_2C_2$ together with Assumption \ref{ass:rank} yields 
 $d_{2,k-1}=M_2C_2(x_k-f(x_{k-1})+G_1Sh_1(x_{k-1})+G_1S(V_1v_{k-1}-z_{1,k-1})-Ww_{k-1})$. This, along with \eqref{eq:d12} and \eqref{eq:d1}, \sy{leads to}

 \vspace{-.3cm}
 {\small
 \begin{align}\label{eq:d_x}
 d_{k-1}\hspace{-.0cm}=\hspace{-.cm}\Phi x_k\hspace{-.1cm}+\hspace{-.1cm}\mj{\kappa(}x_{k-1})\hspace{-.1cm}+\hspace{-.1cm}A_zz_{1,k-1}\hspace{-.1cm}+\hspace{-.1cm}A_vv_{k-1}\hspace{-.1cm} -\hspace{-.1cm}\Phi Ww_{k-1}. 
 \end{align}
 }\vspace{-0.4cm}
 
 \noindent\mo{The input framers in \eqref{eq:d_framers} \mk{are obtained by}} leveraging \cite[Theorem 1]{yang2019sufficient} to compute a decomposition function for the nonlinear function \sy{$\kappa$}, as well as \mo{applying} \cite[Lemma 1]{efimov2013interval} \mo{to bound the linear terms in the right-hand side of \eqref{eq:d_x}}. 
\end{proof}
\vspace*{-0.3cm}
\subsection{ISS and $\mathcal{H}_{\infty}$-Optimal Interval Observer Synthesis} \vspace*{-0.05cm} 
Next, we 
provide sufficient conditions 
to guarantee the
stability of the proposed 
{framers}, i.e.,  
we 
\yo{seek} to synthesize 
{the observer gain $L$ to ensure {input-to-state stability (ISS) of}} 
the observer state error, $e^x_k \triangleq \overline{x}_k-\underline{x}_k$ 
in the sense of Definition \ref{defn:error}, while ensuring that the design is optimal in the sense of \yo{minimizing the  $\mathcal{H}_{\infty}$ gain} (cf. Definition \ref{defn:stability}). 

First, we derive the observer error dynamics as follows. 
\vspace*{-0.2cm}

\begin{lem}\label{lem:error_dyn}
Consider {\color{black}the nonlinear system \eqref{eq:system}} and suppose all assumptions in Theorem \ref{thm:stability} hold. Then, the state framer error dynamics of the resilient interval observer \eqref{eq:framers}--\eqref{eq:d_framers} and its nonlinear comparison system are as follows: 
\begin{align}\label{eq:error_dyn}
\begin{array}{rll}
\hspace{-.4cm}e^{x}_{k+1}&\hspace{-.2cm}=|A-LC_2|e^x_k+\delta^{\rho}_k+|L|\delta^{\psi_2}_k+|\hat{W}|\delta^w\\
&\hspace{-.2cm}+(|V_a-LV_b|-|A-LC_2||\Lambda NV_2|+|\Lambda NV_2|)\delta^v \\
&\hspace{-.2cm}+(|\Lambda|+|D_a-LD_b|-|A-LC_2||\Lambda|)\delta^{\epsilon}\\
&\hspace{-.2cm}\leq (|A-LC_2|+\overline{F}_{\rho}+|L|\overline{F}_{\psi_2})e^x_k+|\hat{W}|\delta^w\\
&\hspace{-.2cm}+(|V_a-LV_b|-|A-LC_2||\Lambda NV_2|+|\Lambda NV_2|)\delta^v \\
&\hspace{-.2cm} +(|\Lambda|+|D_a-LD_b|-|A-LC_2||\Lambda|) \delta^\epsilon,
\end{array}
\end{align}
where $\delta^\zeta_{k} \triangleq \zeta_d(\overline{x}_k,\underline{x}_k)-\zeta_d(\underline{x}_k,\overline{x}_k), \forall \zeta \in \{\psi_2,\rho\}$, $ \delta^s \triangleq \overline{s}-\underline{s}, \forall s \in \{w,v,\epsilon\}$, and $\overline{F}_{\zeta}, \forall \zeta \in \{\psi_2,\rho\}$ are computed through Proposition \ref{prop:tight_decomp}. Moreover, 
\begin{align}
\begin{array}{rl}
V_a &\triangleq A\Lambda NV_2+\Lambda(I-NC_2)G_1SV_1,\\
V_b &\triangleq (C_2\Lambda N-I)V_2, D_a \triangleq A\Lambda, D_b \triangleq C_2\Lambda.
\end{array}
\end{align}\vspace{-.4cm}
\end{lem}
\begin{proof}
It follows from \eqref{eq:framers} that \yo{the dynamics of $e^{\gamma}_k \triangleq \overline{\gamma}_k-\underline{\gamma}_k$ is given by} 
$e^{\gamma }_{k+1}=|A-LC_2| e^{\gamma}_k+\delta^\rho_k+|L| \delta^{\psi_{2}}_k+\hspace{-.1cm}|\hat{V}| \delta^{v}+|\hat{W}| \delta^w+|D|\delta^\epsilon$.
\mk{This, combined with} $e^x_k=e^{\gamma}_k+|\Lambda|\delta^\epsilon+|\Lambda NV_2|\delta^v$
\mk{(followed from \eqref{eq:framers_2})} 
 \yo{results in} 
the equality in \eqref{eq:error_dyn}, 
\mk{which together with} 
the facts that $\delta^\zeta_k \leq \overline{F}_\zeta e^x_k, \forall \zeta \in \{\rho,\psi_2\}$ (cf. Proposition \ref{prop:tight_decomp}), yields the inequality in \eqref{eq:error_dyn}.
\end{proof}
\vspace{-.15cm}
Further, by leveraging slightly different approaches to derive an upper \emph{linear} comparison system for the \emph{nonlinear} error comparison system \eqref{eq:error_dyn}, we derive different sets of sufficient conditions to guarantee the ISS property of the proposed observer, as well as to ensure the optimality of the design in the sense of \yo{minimum $\mathcal{H}_{\infty}$ gain}, \sy{as follows}. 
\begin{thm}[{ISS \& $\mathcal{H}_{\infty}$-Optimal Resilient Interval Observer Synthesis}]\label{thm:stability}
Consider 
system \eqref{eq:system} (equivalently the transformed system \eqref{eq:stateq}) 
and suppose Assumptions \ref{assumption:mix-lip}--\ref{ass:abs_inv} hold. Moreover, suppose there exist matrices $\mathbb{R}^{n \times n} \ni P^* \succ \mathbf{0}_{n,n}, {\Gamma^* \in \mathbb{R}^{n \times (l-p_{H})}_{\geq 0}}$ and {$\eta^* \in \mathbb{R}_{>0}$} such that $-P^* \in \mathbb{M}_n$ and the tuple $(P^*,\Gamma^*,\eta^*)$ solves the following problem:
\begin{align}\label{eq:sdp_DT}
\begin{array}{rl}
&\hspace{-.4cm}\min\limits_{\{\eta, P,\Gamma\}} \eta \\
&\hspace{-.4cm}s.t.  \begin{bmatrix} P & P\tilde{A}-\Gamma\tilde{C} & P\tilde{B} - \Gamma\tilde{D}  & 0 \\
  * & P & 0 & I \\
 * & * & \eta I & 0 \\
 * & * & * & \eta I \end{bmatrix}\hspace{-.1cm} \succ \hspace{-.1cm} 0,  (P,\Gamma) \in \mathbf{C},\\
\end{array}
\end{align}
where the matrices $\tilde{A},\tilde{B},\tilde{C},\tilde{D}$, as well as the corresponding additional set of constraints \yo{$\mathbf{C}$} 
can be either of the following:
\begin{enumerate}[(i)]
\item \yo{$\mathbf{C}\hspace{-0.05cm}=\hspace{-0.05cm}\{(P,\Gamma) \mid P \begin{bmatrix} {A} & V_a & D_a \end{bmatrix}\hspace{-0.05cm}-\hspace{-0.05cm}\Gamma \begin{bmatrix}{C}_2 & V_b & D_b \end{bmatrix} \geq 0\}$,}
if:\label{case:D_1}
\begin{align*}
\tilde{A}&= A+\overline{F}_{\rho}, \ \tilde{C}= C_2-\overline{F}_{\psi_2},\\
\tilde{B}&= \begin{bmatrix} V_a+(I-A)|\Lambda NV_2| & |\hat{W}| & D_a+(I-A)|\Lambda| \end{bmatrix},\\
\tilde{D}&= \begin{bmatrix} V_b-C_2|\Lambda NV_2| & 0 & D_b-C_2|\Lambda|  \end{bmatrix}.
\end{align*}
\item \yo{$\mathbf{C}=\{(P,\Gamma) \mid \Gamma \begin{bmatrix}{C}_2 & V_b & D_b \end{bmatrix} \geq 0\}$,} 
if \label{case:D_2}
\begin{align*}
\tilde{A} &= |A|+\overline{F}_{\rho}, \ \tilde{C}= -C_2-\overline{F}_{\psi_2},\\
\tilde{B}&= \begin{bmatrix} |V_a|\hspace{-.1cm}+\hspace{-.1cm}(I\hspace{-.1cm}-\hspace{-.1cm}|A|)|\Lambda NV_2| & |\hat{W}| & (I\hspace{-.1cm}-\hspace{-.1cm}|A|)|\Lambda|\hspace{-.1cm}+\hspace{-.1cm}|D_a| \end{bmatrix},\\
\tilde{D} &= \begin{bmatrix} C_2|\Lambda NV_2|-V_b & 0 & C_2|\Lambda|-D_b  \end{bmatrix}.
\end{align*}
\item \mo{$\mathbf{C}=\{\yo{(P,\Gamma) \mid\, } PA-\Gamma C_2 \geq 0$\}}, if: \label{case:D_3}
\begin{align*}
\hspace{-.6cm}\tilde{A} &= A+\overline{F}_{\rho}, \ \tilde{C}= C_2-\overline{F}_{\psi_2}, \ \tilde{D} = \begin{bmatrix} -V_2 & 0 & 0  \end{bmatrix},\\
\hspace{-.6cm}\tilde{B} &= \begin{bmatrix} |\Lambda(I\hspace{-.1cm}-\hspace{-.1cm}NC_2)G_1SV_1|\hspace{-.1cm}+\hspace{-.1cm}|\Lambda NV_2| & |\hat{W}| & |\Lambda| \end{bmatrix}.
\end{align*}
\end{enumerate}
Then, the proposed resilient interval framer \eqref{eq:framers}--\eqref{eq:d_framers} with the corresponding gain $L=({P}^*)^{-1}\Gamma^*$, is a resilient ISS input and state interval observer in the sense of Definition \ref{defn:stability} and also is $\mathcal{H}_{\infty}$-optimal (cf. Definition \ref{defn:H_inf}). \mk{Finally, in any of the above cases, the LMI in \eqref{eq:sdp_DT} is feasible only if 
 the linear comparison system $(\tilde{A},\tilde{B},\tilde{C},\tilde{D})$ is \sy{detectable.}}
\end{thm}
\begin{proof}
We will show that in each of the cases \eqref{case:D_1}--\eqref{case:D_3}, given the corresponding constraint set \mo{$\mathbf{C}$}, a linear comparison system for the observer state error dynamics \eqref{eq:error_dyn} can be computed in the following form:
\begin{align}\label{eq:comparison_error}
e^x_{k+1}\leq (\tilde{A}-L\tilde{C})e^x_k+(\tilde{B}-L\tilde{D})\tilde{w},
\end{align}
\sy{with $\tilde{w} \triangleq \begin{bmatrix} \delta^{v\top} & \delta^{w\top} & \delta^{\epsilon\top} \end{bmatrix}^\top$, where the detectability of the pair ($\tilde{A},\tilde{C}$) is a necessary condition for stabilizing the comparison system}. If this can be shown, then using the results in \cite[Section 9.2.3]{duan2013lmis}, the solution $(P^*,\Gamma^*)$ to the program in \eqref{eq:sdp_DT} returns the optimal observer gain $L^*=(P^*)^{-1}\Gamma^{*}$ for the linear comparison system \eqref{eq:comparison_error}, and hence, for the original error dynamics \eqref{eq:error_dyn} in the \yo{minimum $\mathcal{H}_{\infty}$ gain} sense 
 with {an $\mathcal{H}_\infty$ gain of} $\eta^*$ {(cf. Definition \ref{defn:H_inf})}. This implies that the above linear comparison system \eqref{eq:comparison_error} satisfies the following asymptotic gain (AG) property \cite{sontag1996new}:
\begin{align*}\small
\limsup_{k \to \infty} \|e^x_k\|_{2} \hspace{-.1cm} \leq \hspace{-.1cm} \alpha(\| {\tilde{\delta}}\|_{\ell_{\infty}}), \, \forall e^x_0,\sy{\forall 0\hspace{-.1cm}\le\hspace{-.1cm} {\tilde{\delta}\hspace{-.1cm}\le\hspace{-.1cm}[(\delta^{w})\hspace{-.07cm}^\top \, (\delta^{v})\hspace{-.07cm}^\top \, (\delta^{\epsilon})\hspace{-.07cm}^\top]^\top},} 
\normalsize
\end{align*}
where $\tilde{\delta}$ is any realization of the augmented noise and outer-approximation error {interval} width and $\alpha$ is any class $\mathcal{K}_{\infty}$ function that is lower bounded by $\eta^*\tilde{\delta}$. On the other hand, by setting $\delta=0$, the LMIs in \eqref{eq:sdp_DT} 
{reduce} to their noiseless counterparts in {\cite[Eq. (19)]{moh2022intervalACC}}. Hence, by {\cite[Theorem 2]{moh2022intervalACC}}, the comparison system \eqref{eq:comparison_error} is 0-stable (0-GAS), which in addition to the AG property \mk{above} is equivalent to the ISS property for \eqref{eq:comparison_error} by \cite[Theorem 1-e]{sontag1996new}. Thus, the designed observer is also ISS. So, what remains to complete the proof is to show that the comparison system \eqref{eq:comparison_error} can indeed be computed in each of the cases as follows.

\textbf{Case} \eqref{case:D_1}. Consider the nonlinear comparison system in \eqref{eq:error_dyn}. By satisfying the constraint set $\mathbf{C}$, we enforce {$-P$} to be Metzler, 
  as well as $P\tilde{A}-\Gamma \tilde{C},PV_a-\Gamma V_b$ and $PV_a-\Gamma V_b$ to be non-negative. Also, $\Gamma$ is non-negative by assumption. 
  Consequently, {since {$P$} is positive definite,} 
  it becomes a non-singular M-matrix, {i.e., a square matrix whose negation is Metzler and whose eigenvalues have non-negative real parts,} and hence is inverse-positive \cite[Theorem 1]{plemmons1977m}, i.e., {$P^{-1} \geq 0$}. Therefore, {$L=P^{-1}\Gamma \geq 0$}, $A-LC_2=P^{-1}(PA-\Gamma C_2) \geq 0$, $V_a-LV_b=P^{-1}(PV_a-\Gamma V_b) \geq 0$ and {$D_a-LD_b=P^{-1}(PD_a-\Gamma D_b) \geq 0$}, because they are matrix products of non-negative matrices. \mk{So}, $|L|=L, |A-LC_2|=A-LC_2,|V_a-LV_b|=V_a-LV_b$ and $|D_a-LD_b|=D_a-LD_b$, \mk{which turns \eqref{eq:error_dyn} into the form of \eqref{eq:comparison_error}}. 
  
  \textbf{Case} \eqref{case:D_2}. By applying the triangle inequality, the comparison system in \eqref{eq:error_dyn} can get upper bounded again as

\vspace{-.3cm}
  {\small
  \begin{align}\label{eq:error_dyn_2}
\begin{array}{rll}
\hspace{-.3cm}e^{x}_{k+1}&\hspace{-.2cm}\leq (|A|+|LC_2|+\overline{F}_{\rho}+|L|\overline{F}_{\psi_2})e^x_k+|\hat{W}|\delta^w\\
&\hspace{-.2cm}+(|V_a|\hspace{-.1cm}+\hspace{-.1cm}|LV_b|-|LC_2||\Lambda NV_2|\hspace{-.1cm}+\hspace{-.1cm}(I-|A|)|\Lambda NV_2|)\delta^v \\
&\hspace{-.2cm} +((I-|A|)|\Lambda|+|D_a|+|LD_b|-|LC_2||\Lambda|) \delta^\epsilon.
\end{array}
\end{align}
}\vspace{-0.2cm}

\noindent By a similar argument as in Case \eqref{case:D_1}, enforcing $-P$ to be \yo{Metzler} along with the constraints set $\mathbf{C}$ results in $|LC_2|=LC_2,|LV_b|=LV_b$ and $|LD_b|=LD_b$, \mo{and hence} turns \eqref{eq:error_dyn_2} into the form of \eqref{eq:comparison_error}.

  \textbf{Case} \eqref{case:D_3}. Note that by the triangle inequality,
$|V_a-LV_b|=|(A-LC_2)\Lambda NV_2+LV_2+\Lambda(I-NC_2)G_1SV_1|\leq |(A-LC_2)||\Lambda NV_2|+|L||V_2|+|\Lambda (I-NC_2)G_1SV_1|$, and
$|D_a-LD_b|=|(A-LC_2)\Lambda| \leq |(A-LC_2)||\Lambda|$. These two combined with \eqref{eq:error_dyn} yield
\begin{align}\label{eq:error_dyn_3}\small
\begin{array}{rl}
\hspace{-.2cm}e^{x}_{k+1}&\hspace{-.2cm}\leq (|A\hspace{-.1cm}-\hspace{-.1cm}LC_2|\hspace{-.1cm}+\hspace{-.1cm}\overline{F}_{\rho}\hspace{-.1cm}+\hspace{-.1cm}|L|\overline{F}_{\psi_2})e^x_k\hspace{-.1cm}+\hspace{-.1cm}|\hat{W}|\delta^w\hspace{-.1cm}+\hspace{-.1cm}|\Lambda| \delta^\epsilon\\
&\hspace{-.2cm}+(|L||V_2|\hspace{-.1cm}+\hspace{-.1cm} |\Lambda NV_2|\hspace{-.1cm}+\hspace{-.1cm}|\Lambda(I-NC_2)G_1SV_1|)\delta^v.
\end{array}\normalsize
\end{align}
The rest of the proof is to enforce that $A-LC_2$ and $L$ are non-negative to turn \eqref{eq:error_dyn_3} into the form of \eqref{eq:comparison_error}, which \mk{is similar to} the 
the proofs of the previous two cases. 
\end{proof}

\section{Illustrative Example} \label{sec:examples}
We now illustrate the effectiveness of our proposed \yo{resilient} observer using 
a three-area \yo{power} system \sy{\cite[Figure 1]{yong2016tcps}}, 
where each control area consists of a generator and load buses 
with transmission lines between areas. The nonlinear continuous-time model of the buses is slightly modified based on \cite{kim2016attack}, with \yo{the subscript} $i$ being the bus number: 
\begin{align*}\small
&\begin{array}{rl}
 \dot{f}_1(t) &\hspace{-.cm}= \hspace{-.cm}-\frac{1}{m_1}(\phi_i(t)- (P_{M_1}(t) + d_{ 1}(t))) +  w_{2, 1}(t),\\
 \dot{f}_i(t) &\hspace{-.cm}= \hspace{-.cm}-\frac{1}{m_i}(\phi_i(t)- P_{M_i}(t))+ \hspace{-.cm}w_{2, i}(t), \ i \in \{2,3\}, \\
    \dot{\theta}_i(t) &= f_i(t) + w_{1, i}(t), \quad \quad \quad \quad \quad \quad \ i \in \{1,2,3\},
 \end{array} \normalsize
\end{align*}
\yo{with $\phi_i(t) \triangleq D_i f_i(t) + \sum_{l\in S_i} P_{il}(t) +P_{L_i}(t))$,} where
$\theta_i$ is the phase angle, $f_i$ is the angular frequency, $m_i = 0.01$, $D_i = 0.11$, $P_{M_i}(t)$ is the mechanical power
(the control input), $P_{L_i}(t)$ is a known power demand, 
$S_i$ is the set of neighboring buses of $i$, and \yo{the nonlinear tie line power flow equation is as follows:} 
 $P_{il}(t) = -P_{li}(t) = t_{il}\sin(\theta_i(t) - \theta_l(t))$, 
with ${\color{black}t_{il}=1}$. \yo{Only the actuator of Control Area 1 is attacked and the false data injection signal is $d_1(t)$.}

\yo{On the other hand, the output equation is given as follows: 
\begin{align*}\small
y_i (t)&= [\theta_{i}(t) \ f_{i}(t)]^\top + v_i(t), \ i \in\{1,3\},\\
y_2(t) &= [\theta_2(t) \ f_2(t)]^\top + d_2 (t)+ v_2(t),\normalsize
\end{align*}
where only the sensor $y_2(t)$ is injected with a false data signal $d_2(t)$. Thus, the concatenated attack/unknown input signal is $d(t)=[d_1(t) \ d_2(t)]^\top$ and the $G$ and $H$ matrices in \eqref{eq:system}   corresponding to the attack locations are given by $G = \begin{bmatrix}
     0 & 0 & 0 & 0 & 0 & 0 \\
     0 & 1 & 0 & 0 & 0 & 0 
 \end{bmatrix}^\top$ and  $H = \begin{bmatrix}
     0 & 0 & 0 & 0 & 0 & 0 \\
     0 & 0 & 0 & 1 & 0 & 0 
 \end{bmatrix}^\top$.}


 In our simulations, \yo{the \emph{forward Euler} method is used to discretize the system dynamics with a sampling time $dt = 0.01 s$ and} 
 both $P_{M_i}(t)$ and $P_{L_i}(t)$ were set to be identically zero. Moreover, for $i=1,\dots,3$, the process noise $w_{i}(t)$ and the measurement noise $v_{i}(t)$ were assumed to be bounded within the bounds $\left[\begin{bmatrix} -50 & -50 \end{bmatrix}^\top , \begin{bmatrix} 50  & 50 \end{bmatrix}^\top\right]$ and $\left[\begin{bmatrix} -0.5 & -0.5 \end{bmatrix}^\top , \begin{bmatrix} 0.5 & 0.5 \end{bmatrix}^\top\right]$, respectively. 

For the sake of comparison, we first applied our previous input and state observer \cite{khajenejad2020simultaneous} \yo{that does not have stabilizing gains} to the above example, which \yo{we found to not be able to yield stable interval estimates (i.e., the framer interval width diverges).} 
On the other hand, \mo{\yo{when} implementing the proposed observer in \eqref{eq:framers}--\eqref{eq:d_framers}, the \yo{optimization problem} 
in \eqref{eq:sdp_DT} was solved with the additional linear constraints in Case \eqref{case:D_3}, \yo{and we obtained} 
the \mk{following} observer gain:}

\vspace{-.3cm}
{\small
\begin{align*}
    L = \begin{bmatrix}
         0.70  & 0 & 0.27 & 0 &  0 \\
        0 &   0 &   0.38 &  0 &   0 \\
       0 &    0.83  & 73.19 &  0 & 0 \\
       -0.0022 &    0.0084 & 174.55  &  0.0056 &  -0.0001 \\
       0  &  0 &  0.14 &    0.70 &    0.005 \\
        0.0050  &  0.0098  &  0.11 &   0.01 &    0.62
    \end{bmatrix}.
\end{align*}
}\vspace{-0.1cm}

 As shown in Figures \ref{fig:states} and \ref{fig:attacks}{\color{black},
all the states and attack signals are bounded by the framers \yo{computed} 
by the proposed observer, demonstrating its \mo{correctness}} \yo{and ability to obtain resilient state estimates and to reconstruct attack signals}. Finally, \mo{as} shown in Figure \ref{fig:error}, the actual state and input estimation error sequences \yo{(i.e., the framer interval widths)} converge to steady-state values, \yo{demonstrating} the input-to-state stability of the \yo{proposed interval} observer. 

\begin{figure}[t]
\centering

\includegraphics[width=0.225\textwidth,trim=0mm 0mm 0mm 0mm,clip]{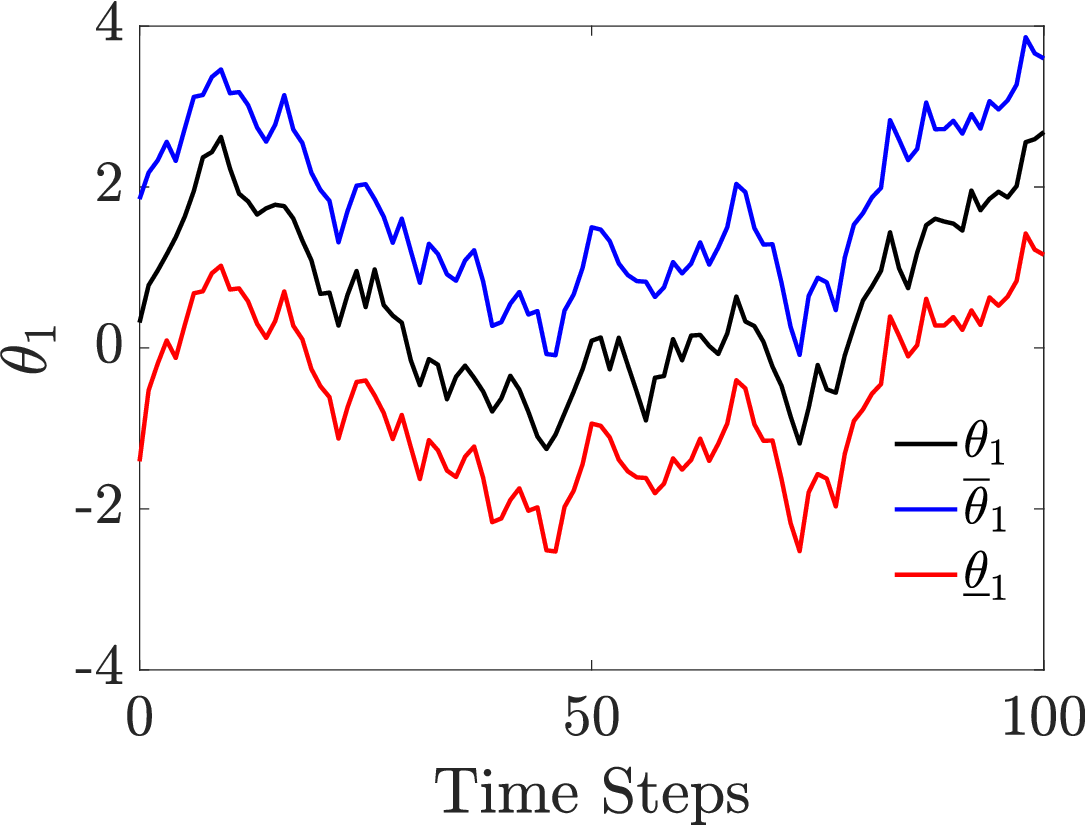} 
\includegraphics[width=0.225\textwidth,trim=0mm 0mm 0mm 0mm,clip]{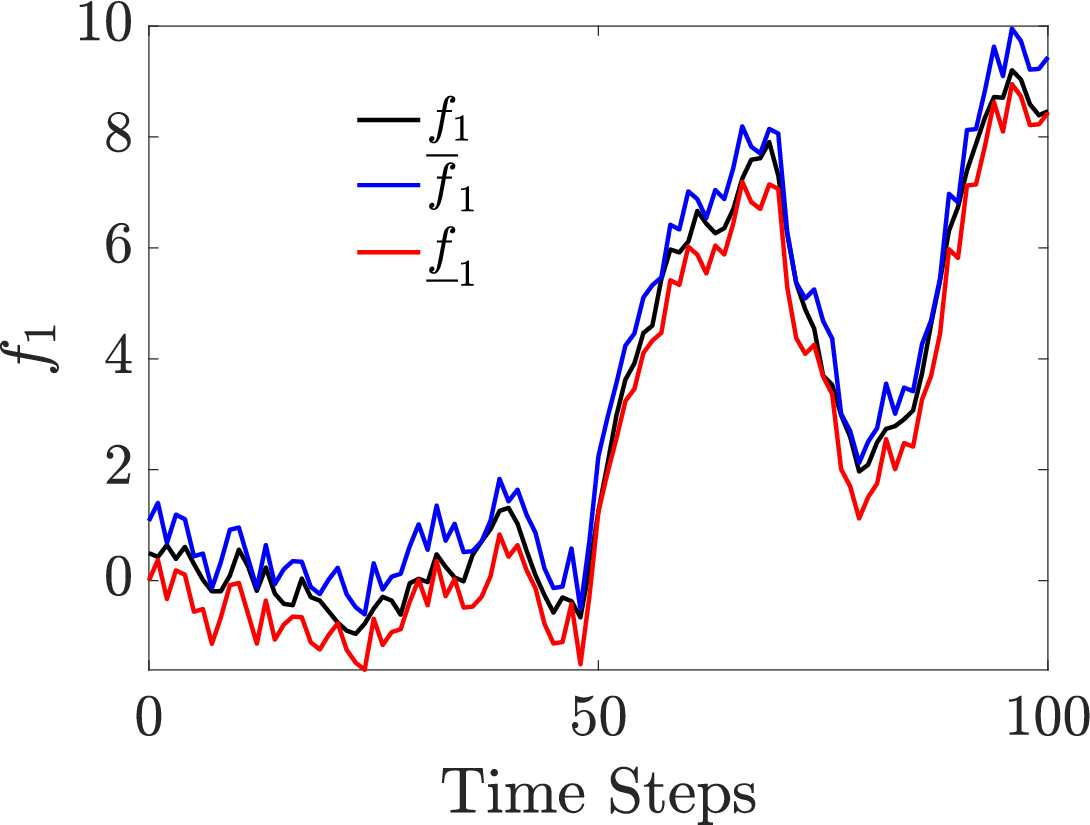} 

\includegraphics[width=0.225\textwidth,trim=0mm 0mm 0mm 0mm,clip]{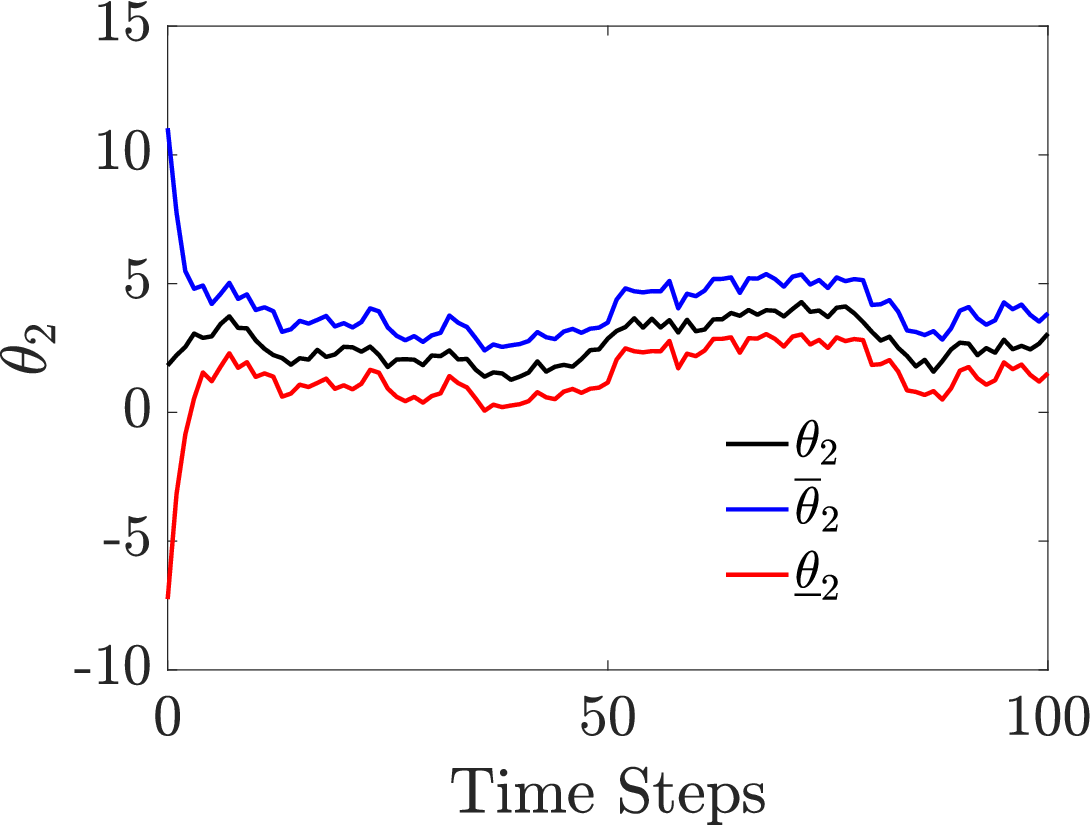} 
\includegraphics[width=0.225\textwidth,trim=0mm 0mm 0mm 0mm,clip]{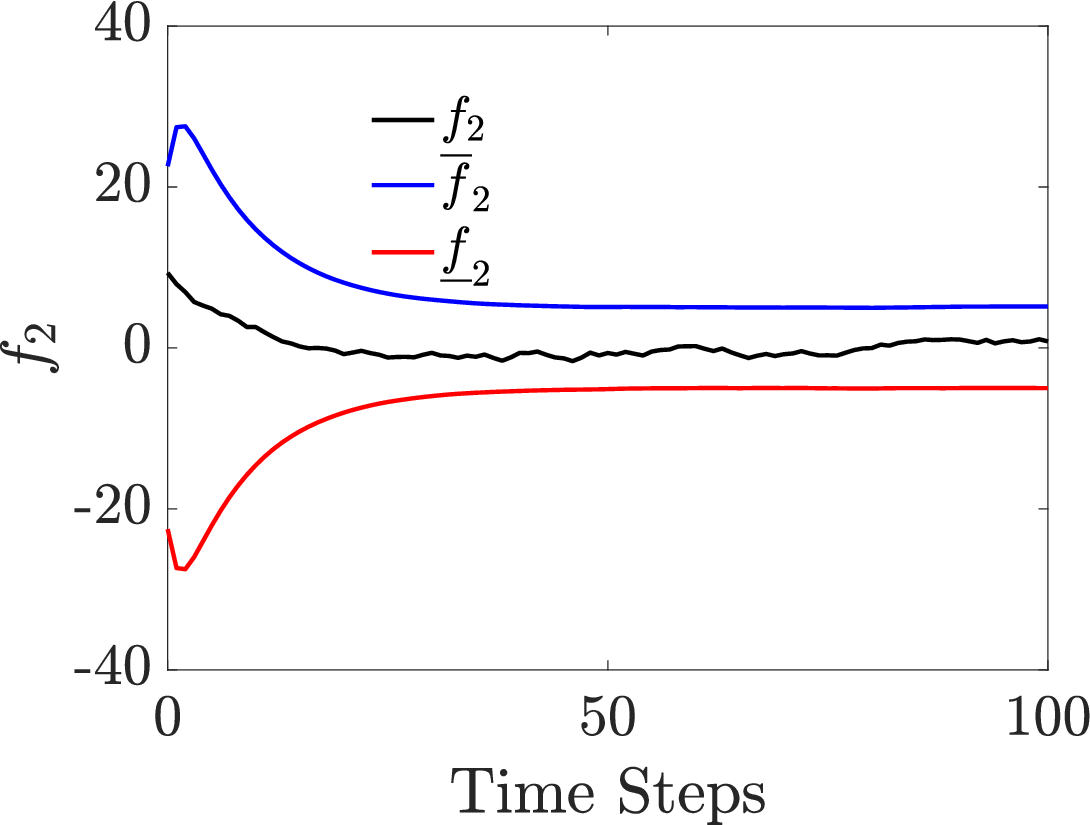} 

\includegraphics[width=0.225\textwidth,trim=0mm 0mm 0mm 0mm,clip]{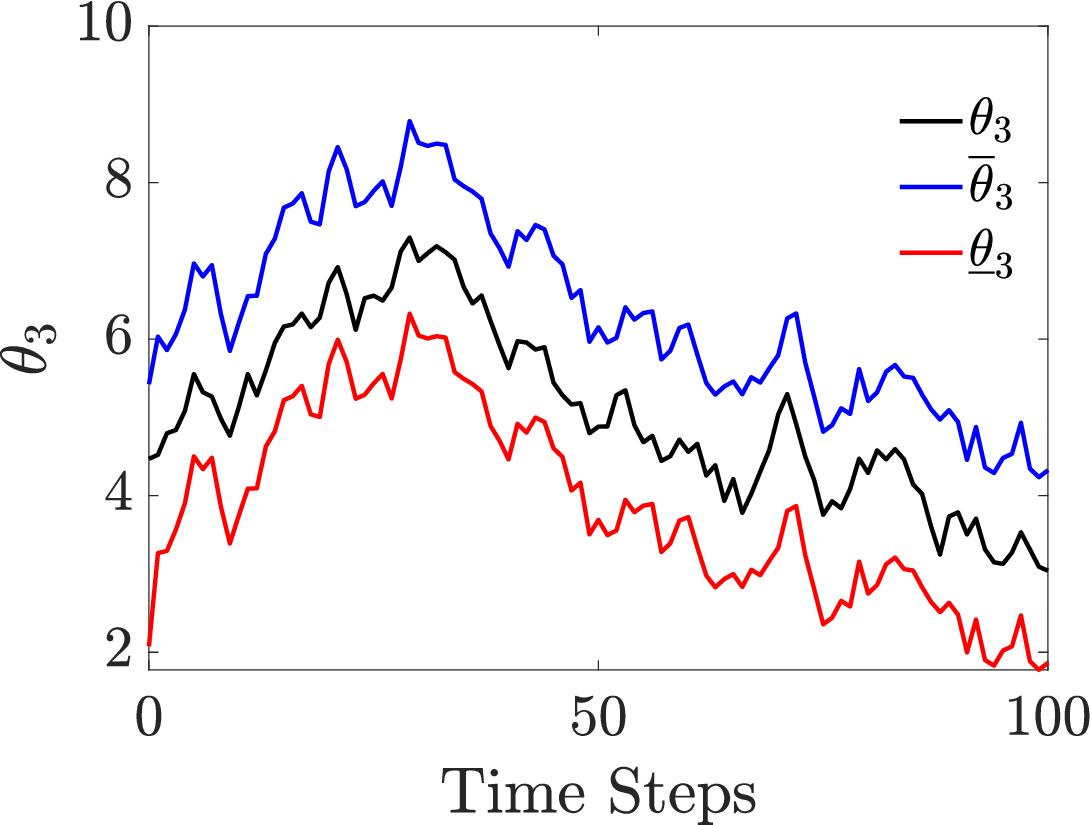} 
\includegraphics[width=0.225\textwidth,trim=0mm 0mm 0mm 0mm,clip]{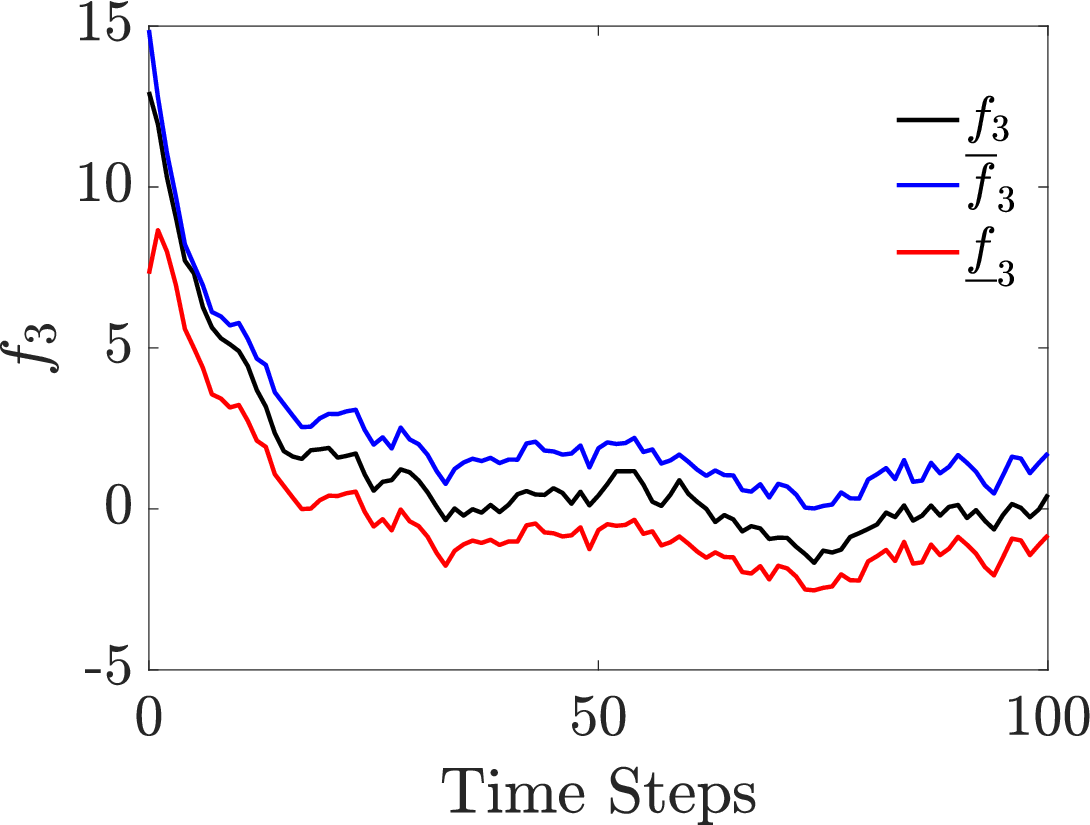} 
\caption{{States: $\theta_i,f_i$, and their upper and lower framers $\overline{\theta}_i,\underline{\theta}_1,\overline{f}_i,\underline{f}_i$, returned by the proposed approach.}}
\label{fig:states}
\vspace{-0.05cm}
\end{figure}

\begin{figure}[t]
\centering
\includegraphics[width=0.235\textwidth,trim=2mm 0mm 5mm 0mm,clip]{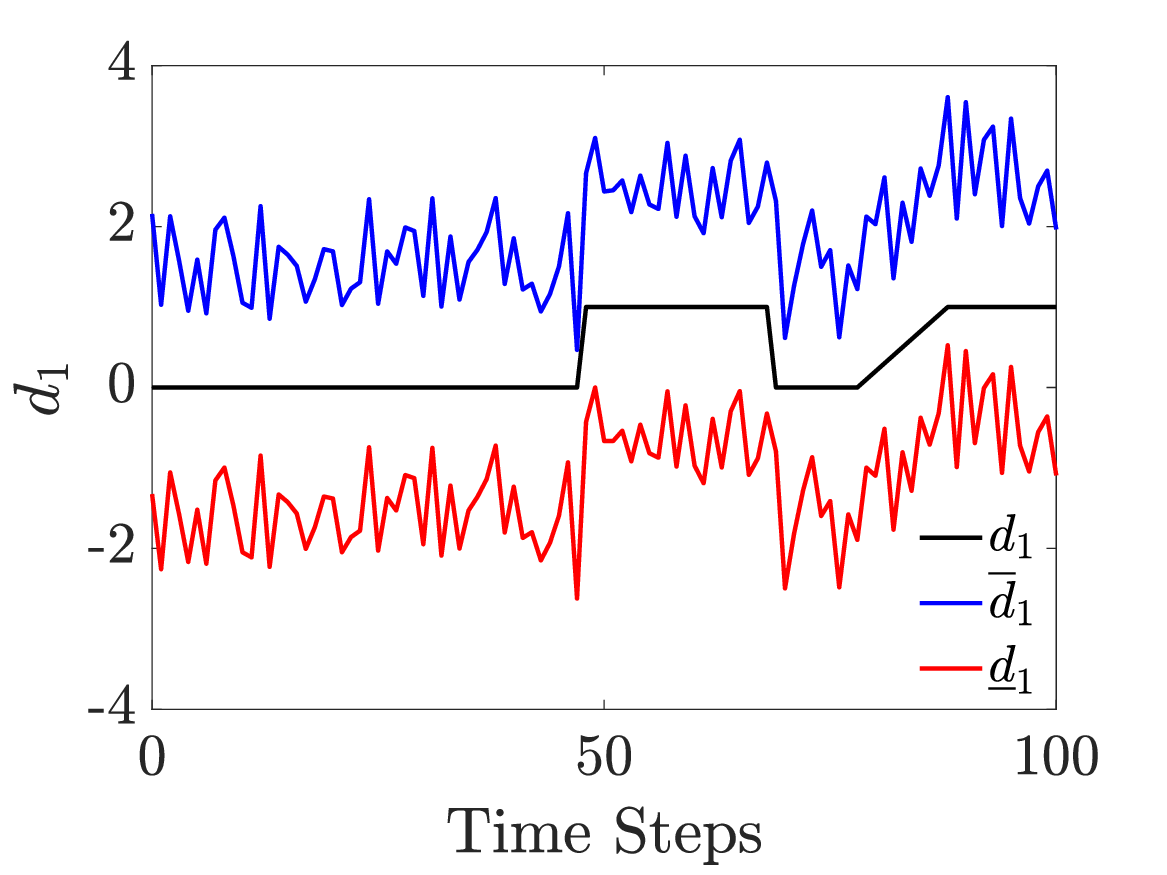} 
\includegraphics[width=0.235\textwidth,trim=0mm 0mm 0mm 0mm,clip]{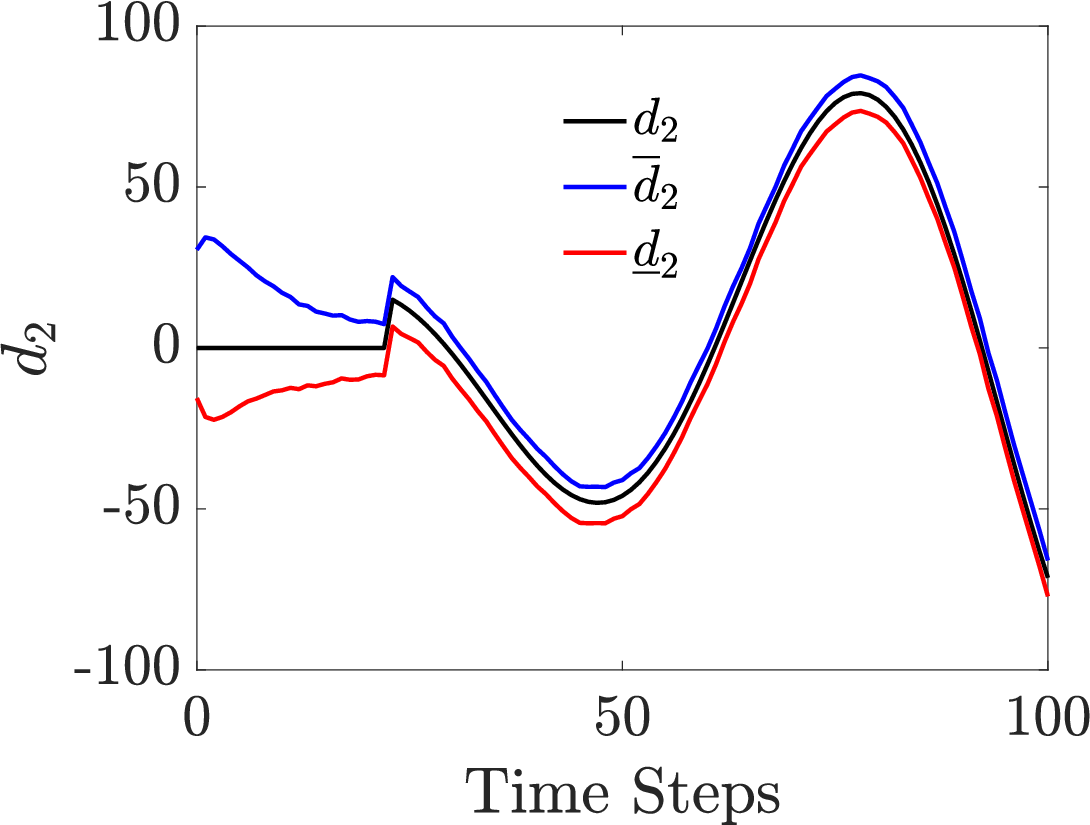} 
\caption{{Attack signals: $d_1,d_2$, and their upper and lower framers $\overline{d}_1,\underline{d}_1,\overline{d}_2,\underline{d}_2$, returned by the proposed approach.}}
\label{fig:attacks}\vspace{-0.05cm}
\end{figure}

\begin{figure}[t]
\centering
\includegraphics[width=0.239\textwidth,trim=0mm 0mm 8mm 0mm,clip]{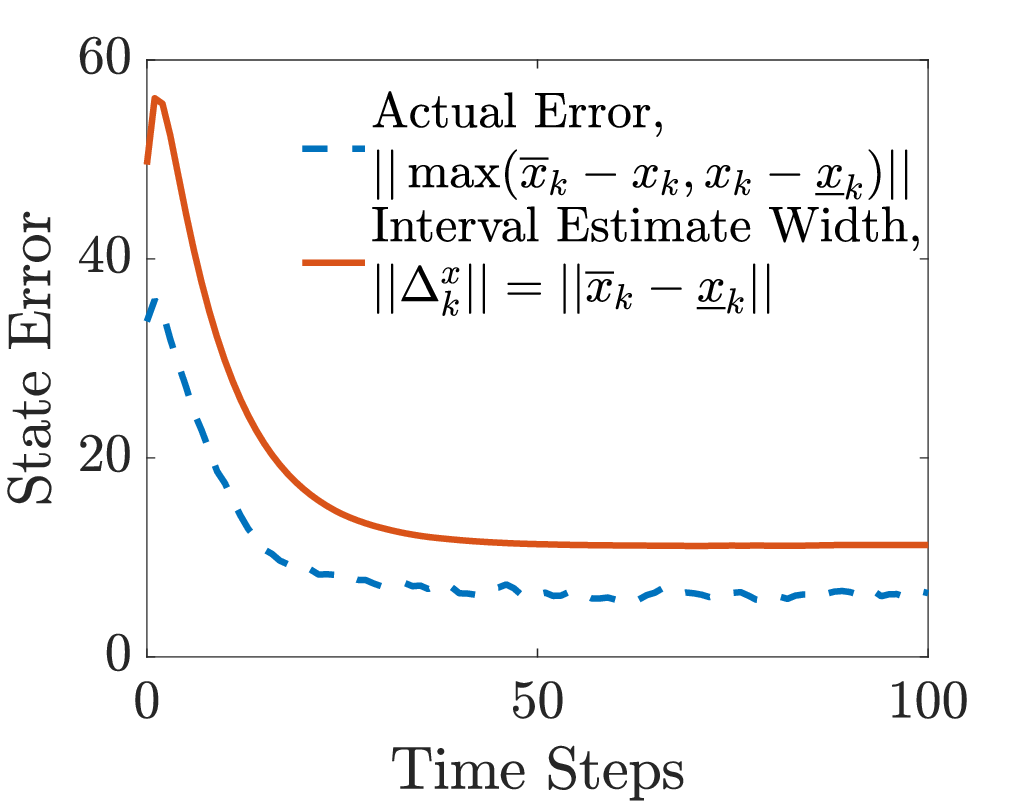} 
\includegraphics[width=0.239\textwidth,trim=0mm 0mm 8mm 0mm,clip]{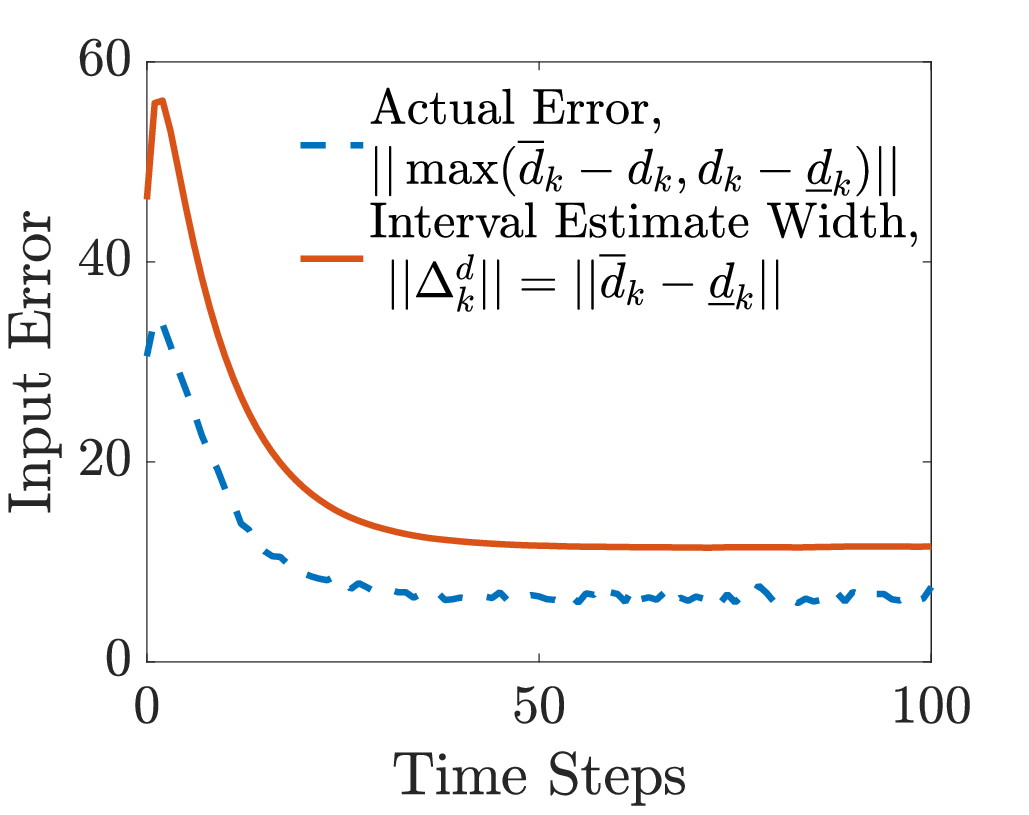} 
\caption{State and input estimation error sequences.
}
\label{fig:error} \vspace{-0.05cm}
\end{figure}

\vspace{-0.1cm}
\section{Conclusion} \label{sec:conclusion}
\vspace{-0.05cm}
In this paper, the problem of \mo{resilient state estimation and attack reconstruction for nonlinear discrete-time systems with nonlinear observations/constraints, that are subject to 
bounded noise signals, was addressed. In the considered setting, both sensors and actuators \sy{could be} affected by attack signals/unknown inputs.}
By introducing auxiliary states, as well as taking advantage of mixed-monotone decomposition of nonlinear \mo{functions} and affine \mo{parallel} outer-approximation of the observation functions, the proposed observer was shown to be correct, i.e.,  it  
recursively computes interval estimates that by construction, contain the true states and unknown inputs of the system. Further, several semi-definite programs 
were provided to synthesize the proposed observer gains \mo{that guarantee} input-to-state stability \mo{of the observer} and \mo{optimality} of the \yo{proposed interval observer} design. \yo{Future work will include alternative designs for minimizing $L_1$ gain, similar to \cite{pati2022L1}, as well as an extension to continuous-time nonlinear systems and hybrid systems.} 

\bibliographystyle{ieeetran}
{\tiny
\bibliography{biblio}
}

\appendices

\end{document}